\documentclass[journal,10pt]{IEEEtran}
\ifCLASSINFOpdf
\else
\fi
\usepackage{comment} 

\usepackage{amsmath,graphicx}

\usepackage{pstricks,graphicx}
\usepackage{setspace}
\usepackage{psfrag}
\usepackage{amssymb}
\usepackage{amsmath}
\usepackage{amsthm}

\newtheorem{remark}{Remark}
\newtheorem{assumption}{Assumption}
\newtheorem{lemma}{Lemma}
\newtheorem{corollary}{Corollary}
\newtheorem{theorem}{Theorem}
\newtheorem{proposition}{Proposition}
\newtheorem{condition}{Condition}
\newtheorem{example}{Example}
\newtheorem{definition}{Definition}

\begin{document}



\title{Distributed Optimization Using the Primal-Dual Method of Multipliers}

\graphicspath{{figures/}}
%

\author{Guoqiang~Zhang and Richard Heusdens
\thanks{G.~Zhang is with both the School of Computing and Communications, University of Technology, Sydney, Australia, and the Department
of Microelectronics, Circuits and Systems group, Delft University of Technology, The Netherlands. Email: {guoqiang.zhang@uts.edu.au}}

\thanks{R.~Heusdens is with the Department
of Microelectronics, Circuits and Systems group, Delft University of Technology, The Netherlands.
Email: {r.heusdens@tudelft.nl}}
\thanks{Part of the work has been published on ICASSP, 2015, with the paper titled \emph{Bi-Alternating Direction Method of Multipliers over Graphs}. After careful consideration, we decide to change the name of our algorithm from \emph{bi-alternating direction method of multipliers (BiADMM)} in \cite{xiaoqiang14BiADMM} and \cite{xiaoqiang15BiADMM} to \emph{primal-dual method of multipliers (PDMM)}.}
}

\maketitle

\begin{abstract}
In this paper, we propose the primal-dual method of multipliers (PDMM) for distributed optimization over a graph. In particular, we optimize a sum of convex functions defined over a graph, where every edge in the graph carries a linear equality constraint. In designing the new algorithm, an augmented primal-dual Lagrangian function is constructed which smoothly captures the graph topology. It is shown that a saddle point of the constructed function provides an optimal solution of the original problem. Further under both the synchronous and asynchronous updating schemes, PDMM has the convergence rate of $O(1/K)$ (where $K$ denotes the iteration index) for general closed, proper and convex functions. Other properties of PDMM such as convergence speeds versus different parameter-settings and resilience to transmission failure are also investigated through the experiments of distributed averaging. 

\end{abstract}


\begin{IEEEkeywords}
Distributed optimization, ADMM, PDMM, sublinear convergence.
\end{IEEEkeywords}

%
\IEEEpeerreviewmaketitle

\section{Introduction}

In recent years, distributed optimization has drawn increasing attention due to the demand for big-data processing and easy access to ubiquitous computing units (e.g., a computer, a mobile phone or a sensor equipped with a CPU). The basic idea is to have a set of computing units collaborate with each other in a distributed way to complete a complex task. Popular applications include telecommunication \cite{Richardson08Coding,xiaoqiang13ADMMLDPC}, wireless sensor networks \cite{Boyd06gossip}, cloud computing and machine learning \cite{Sontag11ML}. The research challenge is on the design of efficient and robust distributed optimization algorithms for those applications.

To the best of our knowledge, almost all the optimization problems in those applications can be formulated as optimization over a graphic model $G=(\mathcal{V},\mathcal{E})$:
\begin{align}
\min_{\{\boldsymbol{x}_i\}}\sum_{i\in \mathcal{V}}f_i(\boldsymbol{x}_i)+\sum_{(i,j)\in \mathcal{E}}f_{ij}(\boldsymbol{x}_i,\boldsymbol{x}_j),\label{equ:optProGen}
\end{align}
where $\{f_i|i\in \mathcal{V}\}$ and $\{f_{ij}|(i,j)\in \mathcal{E}\}$ are referred to as node and edge-functions, respectively. For instance, for the application of distributed quadratic optimization, all the node and edge-functions are in the form of scalar quadratic functions (see \cite{xiaoqiang12LiCDQO,Moallemi09GaBP,xiaoqiang14MSM}). 

In the literature, a large number of applications (see \cite{Boyd11ADMM}) require that every edge function $f_{ij}(\boldsymbol{x}_i,\boldsymbol{x}_j)$, $(i,j)\in \mathcal{E}$, is essentially a linear equality constraint in terms of $\boldsymbol{x}_i$ and $\boldsymbol{x}_j$. Mathematically, we use $\boldsymbol{A}_{i j}\boldsymbol{x}_i+\boldsymbol{A}_{j i}\boldsymbol{x}_j=\boldsymbol{c}_{ij}$ to formulate the equality constraint for each $(i,j)\in \mathcal{E}$, as demonstrated in Fig.~\ref{fig:prob_graph}. In this situation,  (\ref{equ:optProGen}) can be described as
\begin{align}
\min_{\{\boldsymbol{x}_i\}}\sum_{i\in \mathcal{V}}f_i(\boldsymbol{x}_i)+\sum_{(i,j)\in \mathcal{E}}I_{\boldsymbol{A}_{i j}\boldsymbol{x}_i+\boldsymbol{A}_{j i}\boldsymbol{x}_j=\boldsymbol{c}_{ij}}(\boldsymbol{x}_i,\boldsymbol{x}_j),\label{equ:optProMulti}
\end{align}
where $I_{(\cdot)}$ denotes the indicator or characteristic function defined as $I_\mathcal{C}(\boldsymbol{x})=0$ if $\boldsymbol{x}\in \mathcal{C}$ and $I_\mathcal{C}(\boldsymbol{x})=\infty$ if $\boldsymbol{x} \notin \mathcal{C}$.
In this paper, we focus on convex optimization of form (\ref{equ:optProMulti}),  where every node-function $f_i$ is closed, proper and convex.

\begin{figure}[tb]
\centering
\begin{footnotesize}
  \includegraphics[width=70mm]{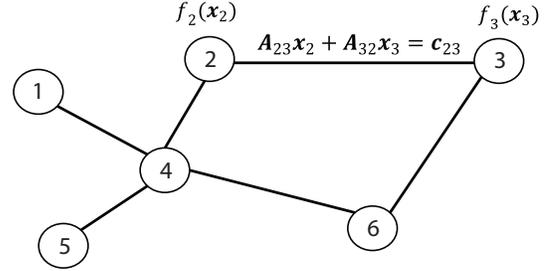}
\end{footnotesize}
\caption{\small Demonstration of Problem~(\ref{equ:optProGen}) for edge-functions being linear constraints. Every edge in the graph carries an equality constraint. 
} \label{fig:prob_graph}
\end{figure}

The majority of recent research have been focusing on a specialized form of the convex problem (\ref{equ:optProMulti}), where every edge-function $f_{ij}$ reduces to $I_{\boldsymbol{x}_i=\boldsymbol{x}_j}(\boldsymbol{x}_i,\boldsymbol{x}_j)$. The above problem is commonly known as the \emph{consensus problem} in the literature. Classic methods include the dual-averaging algorithm \cite{Duchi12dualAve}, the subgradient algorithm \cite{Nedic08subgradient}, the diffusion adaptation algorithm \cite{Chen12diffusion}. For the special case that $\{f_i|i\in \mathcal{V}\}$ are scalar quadratic functions (referred to as the \emph{distributed averaging} problem), the most popular methods are the randomized gossip algorithm \cite{Boyd06gossip} and the broadcast algorithm \cite{Iutzeler13gossipAlg}. See \cite{Dimakis10GossipAlg} for an overview of the literature for solving the distributed averaging problem. 

The alternating-direction method of multipliers (ADMM) can be applied to solve the general convex optimization (\ref{equ:optProMulti}). The key step is to decompose each equality constraint $\boldsymbol{A}_{i j}\boldsymbol{x}_i+\boldsymbol{A}_{j i}\boldsymbol{x}_j=\boldsymbol{c}_{ij}$ into two constraints such as $\boldsymbol{A}_{i j}\boldsymbol{x}_i+\boldsymbol{z}_{ij}=\boldsymbol{c}_{ij}$ and $ \boldsymbol{z}_{ij}=\boldsymbol{A}_{j i}\boldsymbol{x}_j$ with the help of the auxiliary variable $\boldsymbol{z}_{ij}$.  As a result, (\ref{equ:optProMulti}) can be reformulated as
\begin{align}
\min_{\boldsymbol{x},\boldsymbol{z}} f(\boldsymbol{x})+g(\boldsymbol{z})\quad \textrm{subject to}\quad \boldsymbol{A}\boldsymbol{x}+\boldsymbol{B}\boldsymbol{z}=\boldsymbol{c}, \label{equ:optProTwo}
\end{align}
where $f(\boldsymbol{x})=\sum_{i\in \mathcal{V}}f_i(\boldsymbol{x}_i)$, $g(\boldsymbol{z})=0$ and $\boldsymbol{z}$ is a vector obtained by stacking up $\boldsymbol{z}_{ij}$ one after another. See \cite{Shi14ADMM} for using ADMM to solve the consensus problem of (\ref{equ:optProMulti}) (with edge-function $I_{\boldsymbol{x}_i=\boldsymbol{x}_j}(\boldsymbol{x}_i,\boldsymbol{x}_j)$). The graphic structure is implicitly embedded in the two matrices $(\boldsymbol{A},\boldsymbol{B})$ and the vector $\boldsymbol{c}$. The reformulation essentially converts the problem on a general graph with many nodes (\ref{equ:optProMulti}) to a graph with only two nodes (\ref{equ:optProTwo}), allowing the application of ADMM. Based on (\ref{equ:optProTwo}), ADMM then constructs and optimizes an augmented Lagrangian function iteratively with respect to $(\boldsymbol{x},\boldsymbol{z})$ and a set of Lagrangian multipliers. We refer to the above procedure as synchronous ADMM as it updates all the variables at each iteration. Recently, the work of \cite{Wei13ADMM} proposed asynchronous ADMM, which optimizes the same function over a subset of the variables at each iteration.    

We note that besides solving (\ref{equ:optProMulti}), ADMM has found many successful applications in the fields of signal processing and  machine learning (see \cite{Boyd11ADMM} for an overview).  For instance, in \cite{Zhang14ADMMAsyn} and \cite{Chang15ADMMAsyn}, variants of ADMM have been proposed to solve a (possibly nonconvex) optimization problem defined over a graph with a star topology, which is motivated from big data applications. The work of \cite{Bianchi16ADMM} considers solving the consensus problem of (\ref{equ:optProMulti}) (with edge-function $I_{\boldsymbol{x}_i=\boldsymbol{x}_j}(\boldsymbol{x}_i,\boldsymbol{x}_j)$) over a general graph, where each node function $f_i$ is further expressed as a sum of two component functions. The authors of \cite{Bianchi16ADMM} propose a new algorithm which includes ADMM as a special case when one component function is zero. In general, ADMM and its variants are quite simple and often provide satisfactory results after a reasonable number of iterations, making it a popular algorithm in recent years.




In this paper, we tackle the convex problem (\ref{equ:optProMulti}) directly instead of relying on the reformulation (\ref{equ:optProTwo}). Specifically, we construct an augmented primal-dual Lagrangian function for (\ref{equ:optProMulti}) without introducing the auxiliary variable $\boldsymbol{z}$ as is required by ADMM. We show that solving (\ref{equ:optProMulti}) is equivalent to searching for a saddle point of the augmented primal-dual Lagrangian. We then propose the primal-dual method of multipliers (PDMM) to iteratively approach one saddle point of the constructed function. It is shown that for both the synchronous and asynchronous updating schemes, the PDMM converges with the rate of $\mathcal{O}(1/K)$ for general closed, proper and convex functions. 

Further we evaluate PDMM through the experiments of distributed averaging. Firstly, it is found that the parameters of PDMM should be selected by a rule (see \ref{subsub:parameterSel}) for fast convergence. Secondly, when there are transmission failures in the graph, transmission losses only slow down the convergence speed of PDMM. Finally, experimental comparison suggests that PDMM outperforms ADMM and the two gossip algorithms in \cite{Boyd06gossip} and \cite{Iutzeler13gossipAlg}.

This work is mainly devoted to the theoretical analysis of PDMM.  In the literature, PDMM has already been successfully applied for solving a few other problems. The work of \cite{Heming15Thesis} investigates the efficiency of ADMM and PDMM for distributed dictionary learning. In \cite{xiaoqiang16BiADMM}, we have used both ADMM and PDMM for training a support vector machine (SVM). In the above examples it is found that PDMM outperforms ADMM in terms of convergence rate. In \cite{Sherson16LCMV_conf}, the authors describes an application of the linearly constrained minimum variance (LCMV) beamformer for use in acoustic wireless sensor networks. The proposed algorithm computes the optimal beamformer output at each node in the network without the need for sharing raw data within the network. PDMM has been successfully applied to perform distributed beamforming. This suggests that PDMM is not only theoretically interesting but also might be powerful in real applications. 


\vspace{-2mm}
\section{Problem Setting}
\label{sec:pre}
In this section, we first introduce basic notations needed in the rest of the paper. We then make a proper assumption about the existence of optimal solutions of the problem. Finally, we derive the dual problem to (\ref{equ:optProMulti}) and its Lagrangian function, which will be used for constructing the augmented primal-dual Lagrangian function in Section~\ref{sec:AugPDLag}.

\vspace{-2mm}
\subsection{Notations and functional properties}
\label{subsec:notation}

We first introduce notations for a graphic model. We denote a graph as $G=(\mathcal{V},\mathcal{E})$, where $\mathcal{V}=\{1,\ldots, m\}$ represents the set of nodes and $\mathcal{E}=\{(i,j)| i, j\in \mathcal{V}\}$ represents the set of edges in the graph, respectively. 
We use $\vec{\mathcal{E}}$ to denote the set of all directed edges. Therefore, $|\vec{\mathcal{E}}|=2|\mathcal{E}|$. The directed edge $[i,j]$ starts from node $i$ and ends with node $j$. We use $\mathcal{N}_i$ to denote the set of all neighboring nodes of node $i$, i.e., $\mathcal{N}_i=\{j|(i,j)\in \mathcal{E}\}$. 
Given a graph $G=(\mathcal{V},\mathcal{E})$, only neighboring nodes are allowed to communicate with each other directly. 

Next we introduce notations for mathematical description in the remainder of the paper. We use bold small letters to denote vectors and bold capital letters to denote matrices. 
The notation $\boldsymbol{M}\succeq 0$ (or $\boldsymbol{M}\succ 0$) represents a symmetric positive semi-definite matrix (or a symmetric positive definite matrix). The superscript $(\cdot)^T$ represents the transpose operator. Given a vector $\boldsymbol{y}$, we use $\|\boldsymbol{y}\|$ to denote its $l_2$ norm. 

Finally, we introduce the conjugate function. Suppose $h:\mathbb{R}^n\rightarrow \mathbb{R}\cup \{+\infty\}$ is a closed, proper and convex function. 
Then the conjugate of $h(\cdot)$ is defined as \cite[Definition 2.1.20]{SawaragiBook85}
\begin{align}
h^{\ast}(\boldsymbol{\delta})\stackrel{\Delta}{=} \max_{\boldsymbol{y}} \boldsymbol{\delta}^T\boldsymbol{y}-h(\boldsymbol{y}), \label{equ:conj_def}
\end{align}
where the conjugate function $h^{\ast}$ is again a closed, proper and convex function. 
Let $\boldsymbol{y}'$ be the optimal solution for a particular $\boldsymbol{\delta}'$ in (\ref{equ:conj_def}). We then  have
\begin{align}
\boldsymbol{\delta}'\in \partial_{\boldsymbol{y}}h(\boldsymbol{y}'),  \label{equ:conj_def2} 
\end{align}
where $\partial_{\boldsymbol{y}}h(\boldsymbol{y}')$ represents the set of all subgradients of $h(\cdot)$ at $\boldsymbol{y}'$ (see \cite[Definition 2.1.23]{SawaragiBook85}).  
As a consequence, since $h^{\ast\ast}=h$, we have 
\begin{align}
h(\boldsymbol{y}')=&\boldsymbol{y}'^T\boldsymbol{\delta}'-h^{\ast}(\boldsymbol{\delta}')=\max_{\boldsymbol{\delta}} \boldsymbol{y}'^T\boldsymbol{\delta}-h^{\ast}(\boldsymbol{\delta}), \label{equ:conj_def3} 
\end{align}
and we conclude that $\boldsymbol{y}'\in \partial_{\boldsymbol{\delta}}h^{\ast}(\boldsymbol{\delta}')$ as well. 

\vspace{-2mm}
\subsection{Problem assumption}
\label{subsec:proAss}
\vspace{-0mm}

With the notation $G=(\mathcal{V},\mathcal{E})$ for a graph, we first reformulate the convex problem (\ref{equ:optProMulti}) as
\begin{align}
\hspace{-0.4mm}\min_{\boldsymbol{x}}\sum_{i\in \mathcal{V}} f_i(\boldsymbol{x}_i) \textrm{ s.~t. }  \boldsymbol{A}_{i j}\boldsymbol{x}_i\hspace{-0.5mm}+\hspace{-0.6mm}\boldsymbol{A}_{j i}\boldsymbol{x}_j\hspace{-0.4mm}=\hspace{-0.4mm}\boldsymbol{c}_{ij} \textrm{ }\forall (i,j)\hspace{-0.4mm}\in\hspace{-0.4mm} \mathcal{E}, \label{equ:optProMulti_re}
\end{align}
where each function $f_i: \mathbb{R}^{n_i}\rightarrow \mathbb{R}\cup \{+\infty\} $ is assumed to be closed, proper and convex, and $\boldsymbol{x}=[\boldsymbol{x}_1^T,\boldsymbol{x}_2^T,\ldots,\boldsymbol{x}_m^T]^T$. For every edge $(i,j)\in \mathcal{E}$, we let $(\boldsymbol{c}_{ij},\boldsymbol{A}_{i j}, \boldsymbol{A}_{j i})\in (\mathbb{R}^{n_{ij}},\mathbb{R}^{n_{ij}\times n_i},\mathbb{R}^{n_{ij}\times n_j})$. The vector $\boldsymbol{x}$ is thus of dimension $n_{\boldsymbol{x}}=\sum_{i\in \mathcal{V}} n_i$. In general, $\boldsymbol{A}_{ij}$ and $\boldsymbol{A}_{ji}$ are two different matrices. The matrix $\boldsymbol{A}_{ij}$ operates on $\boldsymbol{x}_i$ in the linear constraint of edge $(i,j)\in\mathcal{E}$. The notation $\textrm{s.~t.}$ in (\ref{equ:optProMulti_re}) stands for ``subject to". We take the reformulation (\ref{equ:optProMulti_re}) as the \emph{primal} problem in terms of $\boldsymbol{x}$.

The primal Lagrangian for (\ref{equ:optProMulti_re}) can be constructed as
\begin{align}
\hspace{-1.2mm}L_p(\boldsymbol{x},\boldsymbol{\delta})\hspace{-0.8mm}= \hspace{-0.5mm}\sum_{i\in \mathcal{V}}\hspace{-0.4mm} f_i(\boldsymbol{x}_i)\hspace{-0.5mm}+\hspace{-2mm}\sum_{(i,j)\in \mathcal{E}}\hspace{-1.5mm}\boldsymbol{\delta}_{ij}^{T}(\boldsymbol{c}_{ij}\hspace{-0.5mm}-\hspace{-1mm}\boldsymbol{A}_{i j}\boldsymbol{x}_i\hspace{-0.5mm}-\hspace{-1mm}\boldsymbol{A}_{j i}\boldsymbol{x}_j\hspace{-0.3mm}), \label{equ:primalLag}
\end{align}
where $\boldsymbol{\delta}_{ij}$ is the Lagrangian multiplier (or the dual variable) for the corresponding edge constraint in (\ref{equ:optProMulti_re}), and the vector $\boldsymbol{\delta}$ is obtained by stacking all the dual variables $\boldsymbol{\delta}_{ij}$, $(i,j)\in \mathcal{E}$, on top of one another. Therefore, $\boldsymbol{\delta}$ is of dimension $n_{\boldsymbol{\delta}}=\sum_{(i,j)\in \mathcal{E}}n_{ij}$. The Lagrangian function is convex in $\boldsymbol{x}$ for fixed $\boldsymbol{\delta}$, and concave in $\boldsymbol{\delta}$ for fixed $\boldsymbol{x}$. Throughout the rest of the paper, we will make the following (common) assumption:

\begin{assumption}
There exists a saddle point $(\boldsymbol{x}^{\star},\boldsymbol{\delta}^{\star})$ to the Lagrangian function $L_p(\boldsymbol{x},\boldsymbol{\delta})$ such that for all $\boldsymbol{x}\in \mathbb{R}^{n_{\boldsymbol{x}}}$ and $\boldsymbol{\delta}\in \mathbb{R}^{n_{\boldsymbol{\delta}}}$ we have
\begin{eqnarray}
L_p(\boldsymbol{x}^{\star},\boldsymbol{\delta}) \leq L_p(\boldsymbol{x}^{\star},\boldsymbol{\delta}^{\star})\leq L_p(\boldsymbol{x},\boldsymbol{\delta}^{\star}) \nonumber.
\end{eqnarray}
Or equivalently, the following optimality (KKT) conditions hold for $(\boldsymbol{x}^{\star},\boldsymbol{\delta}^{\star})$:
\begin{align}
\sum_{j\in \mathcal{N}_i}\boldsymbol{A}_{i j}^T\boldsymbol{\delta}_{ij}^{\star} \in \partial f_i(\boldsymbol{x}_i^{\star}) &\quad \forall i\in \mathcal{V} \label{equ:KKT_prim_1} \\
\boldsymbol{A}_{ji}\boldsymbol{x}_j^{\star}+\boldsymbol{A}_{ij}\boldsymbol{x}_i^{\star}=\boldsymbol{c}_{ij} &\quad \forall (i,j)\in \mathcal{E}. \label{equ:KKT_prim_2}
\end{align}\vspace{-4mm}
\label{assumption:KKT}
\end{assumption}

\subsection{Dual problem and its Lagrangian function}
\label{subsec:dualPro}
We first derive the dual problem to (\ref{equ:optProMulti_re}). Optimizing $L_p(\boldsymbol{x},\boldsymbol{\delta})$ over $\boldsymbol{\delta}$ and $\boldsymbol{x}$ yields
\begin{align}
&\max_{\boldsymbol{\delta}}\min_{\boldsymbol{x}}L_p(\boldsymbol{x},\boldsymbol{\delta}) \nonumber\\
&=\max_{\boldsymbol{\delta}}\sum_{i\in \mathcal{V}}\min_{\boldsymbol{x}_i}\Big(f_i(\boldsymbol{x}_i)\hspace{-0.5mm}-\hspace{-1.5mm}\sum_{j\in \mathcal{N}_i}\hspace{-1.5mm}\boldsymbol{\delta}_{ij}^T\boldsymbol{A}_{i j}\boldsymbol{x}_i\Big)+\hspace{-1mm}\sum_{(i,j)\in \mathcal{E}}\hspace{-1mm}\boldsymbol{\delta}_{ij}^T\boldsymbol{c}_{ij}\nonumber \\
&=\max_{\boldsymbol{\delta}} \sum_{i\in \mathcal{V}} -f_i^{\ast}\Bigg(\sum_{j\in \mathcal{N}_i}\boldsymbol{A}_{i j}^T\boldsymbol{\delta}_{ij}\Bigg)+\sum_{(i,j)\in \mathcal{E}}\boldsymbol{\delta}_{ij}^T\boldsymbol{c}_{ij},\label{equ:optProMultiDual}
\end{align}
where $f_i^{\ast}(\cdot)$ is the conjugate function of $f_i(\cdot)$ as defined in (\ref{equ:conj_def}), satisfying Fenchel's inequality
\begin{align}
f_i(\boldsymbol{x}_i)+f_i^{\ast}\Bigg(\sum_{j\in \mathcal{N}_i}\boldsymbol{A}_{i j}^T\boldsymbol{\delta}_{ij}\Bigg)\geq \hspace{-1mm}\sum_{j\in \mathcal{N}_i}\boldsymbol{\delta}_{ij}^T\boldsymbol{A}_{i j}\boldsymbol{x}_i. \label{equ:Fenchel_ineq}
\end{align}
Under Assumption~\ref{assumption:KKT}, the dual problem (\ref{equ:optProMultiDual}) is equivalent to the primal problem (\ref{equ:optProMulti_re}). That is suppose $(\boldsymbol{x}^{\star},\boldsymbol{\delta}^{\star})$ is a saddle point of $L_p$. Then $\boldsymbol{x}^{\star}$ solves the primal problem (\ref{equ:optProMulti_re}) and $\boldsymbol{\delta}^{\star}$ solves the dual problem (\ref{equ:optProMultiDual}).

At this point, we need to introduce auxiliary variables to decouple the node dependencies in (\ref{equ:optProMultiDual}). Indeed, every $\boldsymbol{\delta}_{ij}$, associated to edge $(i,j)$, is used by two conjugate functions $f_i^{\ast}$ and $f_j^{\ast}$. As a consequence, all conjugate functions in (\ref{equ:optProMultiDual}) are dependent on each other. To decouple the conjugate functions, we introduce for each edge $(i,j)\in \mathcal{E}$ \emph{two} auxiliary \emph{node} variables   $\boldsymbol{\lambda}_{i|j}\in \mathbb{R}^{n_{ij}}$ and $\boldsymbol{\lambda}_{j|i}\in \mathbb{R}^{n_{ij}}$, one for each node $i$ and $j$, respectively. The node variable $\boldsymbol{\lambda}_{i|j}$ is owned by and updated at node $i$ and is related to neighboring node $j$. Hence, at every node $i$ we introduce $|\mathcal{N}_i|$ new node variables. With this, we can reformulate the original dual problem as
\begin{align}
\max_{\boldsymbol{\delta},\{\boldsymbol{\lambda}_{i}\}}& -\sum_{i\in \mathcal{V}}\hspace{-0.5mm} f_i^{\ast}(\boldsymbol{A}_i^T\boldsymbol{\lambda}_i)+\hspace{-1mm}\sum_{(i,j)\in \mathcal{E}}\hspace{-1mm}\boldsymbol{\delta}_{ij}^T\boldsymbol{c}_{ij}\nonumber \\
&\textrm{ s. t. }\quad \boldsymbol{\lambda}_{i|j}=\boldsymbol{\lambda}_{j|i}=\boldsymbol{\delta}_{ij}\quad \forall (i,j)\in \mathcal{E}, \label{equ:optProMultiDual2}
\end{align}
where $\boldsymbol{\lambda}_i$ is
obtained by vertically concatenating all $\boldsymbol{\lambda}_{i|j}$, $j\in \mathcal{N}_i$, and $\boldsymbol{A}_i^T$ is obtained by horizontally concatenating all $\boldsymbol{A}_{i j}^T$, $j\in \mathcal{N}_i$. To clarify, the product $\boldsymbol{A}_i^T\boldsymbol{\lambda}_i$ in (\ref{equ:optProMultiDual2}) equals to
\begin{align}
\boldsymbol{A}_i^T\boldsymbol{\lambda}_i=\sum_{j\in \mathcal{N}_i}\boldsymbol{A}_{i j}^T\boldsymbol{\lambda}_{i|j}.\label{equ:A_lambda_relation}
\end{align}
Consequently, we let $\boldsymbol{\lambda}=[\boldsymbol{\lambda}_1^T,\boldsymbol{\lambda}_2^T,\ldots,\boldsymbol{\lambda}_m^T]^T$. In the above reformulation (\ref{equ:optProMultiDual2}), each conjugate function $f_i^{\ast}(\cdot)$ only involves the \emph{node} variable $\boldsymbol{\lambda}_i$, facilitating distributed optimization.

Next we tackle the equality constraints in (\ref{equ:optProMultiDual2}). To do so, we construct a (dual) Lagrangian function for the dual problem (\ref{equ:optProMultiDual2}), which is given by
\begin{align}
\hspace{-1.8mm}L_d'(\boldsymbol{\delta},\boldsymbol{\lambda},\boldsymbol{y})\hspace{-0.4mm}=&-\hspace{-0.5mm}\sum_{i\in \mathcal{V}}\hspace{-0.5mm} f_i^{\ast}(\boldsymbol{A}_i^T\boldsymbol{\lambda}_i)+\hspace{-1.5mm}\sum_{(i,j)\in \mathcal{E}}\hspace{-1.5mm}\boldsymbol{\delta}_{ij}^T\boldsymbol{c}_{ij}\nonumber \\
&\hspace{-1.2mm}+\hspace{-2.0mm}\sum_{(i,j)\in \mathcal{E}}\hspace{-1.4mm}\left[\boldsymbol{y}_{i|j}^T(\boldsymbol{\delta}_{ij}\hspace{-0.5mm}-\hspace{-0.6mm}\boldsymbol{\lambda}_{i|j}) \hspace{-0.5mm}+\hspace{-0.6mm}\boldsymbol{y}_{j|i}^T(\boldsymbol{\delta}_{ij}\hspace{-0.5mm}-\hspace{-0.6mm}\boldsymbol{\lambda}_{j|i})\right]\hspace{-0.4mm},
\label{equ:dualLag0}
\end{align}
where $\boldsymbol{y}$ is obtained by concatenating all the Lagrangian multipliers $\boldsymbol{y}_{i|j}$, $[i,j]\in \vec{\mathcal{E}}$, one after another.

We now argue that each Lagrangian multiplier $\boldsymbol{y}_{i|j}$, $[i,j]\in \vec{\mathcal{E}}$, in (\ref{equ:dualLag0}) can be replaced by an affine function of $\boldsymbol{x}_j$. Suppose $(\boldsymbol{x}^{\star},\boldsymbol{\delta}^{\star})$ is a saddle point of $L_{p}$. By letting $\boldsymbol{\lambda}_{i|j}^{\star}=\boldsymbol{\delta}_{ij}^{\star}$ for every $[i,j]\in \vec{\mathcal{E}}$,  Fenchel's inequality (\ref{equ:Fenchel_ineq}) must hold with equality at $(\boldsymbol{x}^{\star},\boldsymbol{\lambda}^{\star})$ from which we derive that
\begin{align}
\boldsymbol{0}&\in \partial_{\boldsymbol{\lambda}_{i|j}}\left[f_i^{\ast}(\boldsymbol{A}_i^{T}\boldsymbol{\lambda}_{i}^{\star})\right]-\boldsymbol{A}_{i j}\boldsymbol{x}_i^{\star}\nonumber \\
&=\partial_{\boldsymbol{\lambda}_{i|j}}\left[f_i^{\ast}(\boldsymbol{A}_{i}^T\boldsymbol{\lambda}_{i}^{\star})\right]+\boldsymbol{A}_{j i}\boldsymbol{x}_j^{\star}-\boldsymbol{c}_{ij}\hspace{3mm} \forall [i,j]\in\vec{\mathcal{E}}.\nonumber
\end{align}
One can then show that $(\boldsymbol{\delta}^{\star},\boldsymbol{\lambda}^{\star},\boldsymbol{y}^{\star})$ where  $\boldsymbol{y}_{i|j}^{\star}\hspace{-0.3mm}=\hspace{-0.3mm}\boldsymbol{A}_{j i}\boldsymbol{x}_j^{\star}\hspace{-0.2mm}-\hspace{-0.2mm}\boldsymbol{c}_{ij}$ for every $[i,j]\in\vec{\mathcal{E}}$, is a saddle point of $L_{d}'$. We therefore restrict the Lagrangian multiplier $\boldsymbol{y}_{i|j}$ to be of the form $\boldsymbol{y}_{i|j}=\boldsymbol{A}_{j i}\boldsymbol{x}_j-\boldsymbol{c}_{ij}$ so that the dual Lagrangian becomes
\begin{align}
\hspace{-1.8mm}L_d(\boldsymbol{\delta},\boldsymbol{\lambda},\boldsymbol{x})\hspace{-0.4mm}=&\hspace{-0.5mm}\sum_{i\in \mathcal{V}}\hspace{-0.5mm}\Big(\hspace{-0.5mm}-\hspace{-0.5mm}f_i^{\ast}(\boldsymbol{A}_i^T\boldsymbol{\lambda}_i)\hspace{-0.5mm}-\hspace{-1.2mm}\sum_{j\in \mathcal{N}_i}\hspace{-1.2mm}\boldsymbol{\lambda}_{j|i}^T(\boldsymbol{A}_{i j}\boldsymbol{x}_i\hspace{-0.2mm}-\hspace{-0.2mm}\boldsymbol{c}_{ij})\Big)\nonumber \\
&-\sum_{(i,j)\in \mathcal{E}}\boldsymbol{\delta}_{ij}^T(\boldsymbol{c}_{ij}-\boldsymbol{A}_{i j}\boldsymbol{x}_i-\boldsymbol{A}_{j i}\boldsymbol{x}_j).
\label{equ:dualLag}
\end{align}

We summarize the result in a lemma below:

\begin{lemma}
If $(\boldsymbol{x}^{\star},\boldsymbol{\delta}^{\star})$ is a saddle point of $L_p(\boldsymbol{x},\boldsymbol{\delta})$, then $(\boldsymbol{\delta}^{\star},\boldsymbol{\lambda}^{\star},\boldsymbol{x}^{\star})$ is a saddle point of $L_d(\boldsymbol{\delta},\boldsymbol{\lambda},\boldsymbol{x})$, where $\boldsymbol{\lambda}_{i|j}^{\star}=\boldsymbol{\delta}_{ij}^{\star}$ for every $[i,j]\in \vec{\mathcal{E}}$.
\label{lemma:P2DSaddlePoint}
\end{lemma}

We note that $L_d(\boldsymbol{\delta},\boldsymbol{\lambda},\boldsymbol{x})$ might not be equivalent to $L_d(\boldsymbol{\delta},\boldsymbol{\lambda},\boldsymbol{y})$. By inspection of the optimality conditions of (\ref{equ:dualLag}), not every saddle point $(\boldsymbol{\delta}^{\star},\boldsymbol{\lambda}^{\star},\boldsymbol{x}^{\star})$ of $L_d$ might lead to $\{\boldsymbol{\lambda}_{i|j}^{\star}=\boldsymbol{\lambda}_{j|i}^{\star},(i,j)\in \mathcal{E}\}$ due to the generality of the matrices $\{\boldsymbol{A}_{ij},[i,j]\in \vec{\mathcal{E}}\}$. In next section we will introduce quadratic penalty functions w.r.t. $\boldsymbol{\lambda}$ to implicitly enforce the equality constraints  $\{\boldsymbol{\lambda}_{i|j}^{\star}=\boldsymbol{\lambda}_{j|i}^{\star},(i,j)\in \mathcal{E}\}$. 


To briefly summarize, one can alternatively solve the dual problem (\ref{equ:optProMultiDual2}) instead of the primal problem. Further, by replacing $\boldsymbol{y}$ with an affine function of $\boldsymbol{x}$ in (\ref{equ:dualLag0}), the dual Lagrangian $L_d(\boldsymbol{\delta},\boldsymbol{\lambda},\boldsymbol{x})$ share two variables $\boldsymbol{x}$ and $\boldsymbol{\boldsymbol{\delta}}$ with the primal Lagrangian $L_p(\boldsymbol{x},\boldsymbol{\delta})$. We will show in next section that the special form of $L_d$ in (\ref{equ:dualLag}) plays a crucial role for constructing the augmented primal-dual Lagrangian.

\vspace{-2mm}
\section{Augmented Primal-Dual Lagrangian}
\label{sec:AugPDLag}
\vspace{-0.5mm}
In this section, we first build and investigate a primal-dual Lagrangian from $L_p$ and $L_d$. We show that a saddle point of the primal-dual Lagrangian does not always lead to an optimal solution of the primal or the dual problem.

To address the above issue, we then construct an \emph{augmented} primal-dual Lagrangian by introducing two additional penalty functions. We show that any saddle point of the augmented primal-dual Lagrangian leads to an optimal solution of the primal and the dual problem, respectively.

\vspace{-2mm}
\subsection{Primal-dual Lagrangian}
\vspace{-1mm}

By inspection of (\ref{equ:primalLag}) and (\ref{equ:dualLag}), we see that in both $L_p$ and $L_d$, the edge variables $\boldsymbol{\delta}_{ij}$ are related to the terms $\boldsymbol{c}_{ij}-\boldsymbol{A}_{i j}\boldsymbol{x}_i-\boldsymbol{A}_{j i}\boldsymbol{x}_j$. As a consequence, if we add the primal and dual Lagrangian functions, the edge variables $\boldsymbol{\delta}_{ij}$ will cancel out and the resulting function contains node variables $\boldsymbol{x}$ and $\boldsymbol{\lambda}$ only. 

We hereby define the new function as the \emph{primal-dual} Lagrangian below: 
\begin{definition}
The primal-dual Lagrangian is defined as 
\begin{align}
&\hspace{-2mm}L_{pd}(\boldsymbol{x},\boldsymbol{\lambda})=L_p(\boldsymbol{x},\boldsymbol{\delta})+L_d(\boldsymbol{\delta},\boldsymbol{\lambda},\boldsymbol{x})
\nonumber \\
&\hspace{-2mm}=\sum_{i\in \mathcal{V}} \hspace{-0.6mm}\Big[f_i(\boldsymbol{x}_i)-\hspace{-1.5mm}\sum_{j\in \mathcal{N}_i}\hspace{-0.5mm}\boldsymbol{\lambda}_{j|i}^T(\boldsymbol{A}_{i j}\boldsymbol{x}_i-\boldsymbol{c}_{ij})\hspace{-0.5mm}-\hspace{-0.5mm}f_i^{\ast}(\boldsymbol{A}_i^T\boldsymbol{\lambda}_{i})\hspace{-0.5mm}\Big].\hspace{-0.2mm}
\label{equ:PDLag}
\end{align}
\end{definition}

$L_{pd}(\boldsymbol{x},\boldsymbol{\lambda})$ is convex in $\boldsymbol{x}$ for fixed $\boldsymbol{\lambda}$ and concave in $\boldsymbol{\lambda}$ for fixed $\boldsymbol{x}$, suggesting that it is essentially a saddle-point problem (see \cite{Chambolle10}, \cite{He2012PD} for solving different saddle point problems). For each edge $(i,j)\in \mathcal{E}$, the node variables $\boldsymbol{\lambda}_{i|j}$ and $\boldsymbol{\lambda}_{j|i}$ substitute the role of the edge variable $\boldsymbol{\delta}_{ij}$. The removal of $\boldsymbol{\delta}_{ij}$ enables to design a distributed algorithm that only involves node-oriented optimization (see next section for PDMM).



Next we study the properties of saddle points of $L_{pd}(\boldsymbol{x},\boldsymbol{\lambda})$: 

\begin{lemma}
If $\boldsymbol{x}^{\star}$ solves the primal problem (\ref{equ:optProMulti_re}), then there exists a $\boldsymbol{\lambda}^{\star}$ such that $(\boldsymbol{x}^{\star},\boldsymbol{\lambda}^{\star})$ is a saddle point of $L_{pd}(\boldsymbol{x},\boldsymbol{\lambda})$.
\label{lemma:saddlePointForward}
\vspace{-4mm}
\end{lemma}
\begin{proof}
If $\boldsymbol{x}^{\star}$ solves the primal problem (\ref{equ:optProMulti_re}), then there exists a $\boldsymbol{\delta}^{\star}$ such that $(\boldsymbol{x}^{\star},\boldsymbol{\delta}^{\star})$ is a saddle point of $L_p(\boldsymbol{x},\boldsymbol{\delta})$ and by Lemma~\ref{lemma:P2DSaddlePoint}, there exist $\boldsymbol{\lambda}_{i|j}^{\star}=\boldsymbol{\delta}_{ij}^{\star}$ for every $[i,j]\in \vec{\mathcal{E}}$ so that $(\boldsymbol{\delta}^{\star},\boldsymbol{\lambda}^{\star},\boldsymbol{x}^{\star})$ is a saddle point of $L_d(\boldsymbol{\delta},\boldsymbol{\lambda},\boldsymbol{x})$. Hence
\begin{align}
L_{pd}(\boldsymbol{x}^{\star},\boldsymbol{\lambda})&=L_p(\boldsymbol{x}^{\star},\boldsymbol{\delta})+L_d(\boldsymbol{\delta},\boldsymbol{\lambda},\boldsymbol{x}^{\star}) \nonumber \\
                     &\leq L_p(\boldsymbol{x}^{\star},\boldsymbol{\delta}^{\star})+L_d(\boldsymbol{\delta}^{\star},\boldsymbol{\lambda}^{\star},\boldsymbol{x}^{\star}) \nonumber   \\
                     &=L_{pd}(\boldsymbol{x}^{\star},\boldsymbol{\lambda}^{\star}) \nonumber \\
                     &\leq L_p(\boldsymbol{x},\boldsymbol{\delta}^{\star})+L_d(\boldsymbol{\delta}^{\star},\boldsymbol{\lambda}^{\star},\boldsymbol{x})\nonumber
\end{align}
\hspace{30mm}$=L_{pd}(\boldsymbol{x},\boldsymbol{\lambda}^{\star}). $
\end{proof}

\vspace{0.5mm}
The fact that $(\boldsymbol{x}^{\star},\boldsymbol{\lambda}^{\star})$ is a saddle point of $L_{pd}(\boldsymbol{x},\boldsymbol{\lambda})$, however, is \emph{not} sufficient for showing $\boldsymbol{x}^{\star}$ (or $\boldsymbol{\lambda}^{\star}$) being optimal for solving the primal problem (\ref{equ:optProMulti_re}) (for solving  the dual problem (\ref{equ:optProMultiDual2})). 

\vspace{-0.5mm}
\begin{example}[$\boldsymbol{x}^{\star}$ not optimal]
Consider the following problem
\begin{align}
\min_{x_1,x_2}f_1(x_1)+f_2(x_2)\quad \textrm{s.t.}\quad x_1-x_2=0,\label{equ:examp1}
\end{align}
\vspace{-3mm}
\begin{align}
\textrm{where }\qquad f_1(x_1)=f_2(-x_1)=\left\{\begin{array}{ll}x_1-1 & x_1\geq 1 \\ 0 & \textrm{otherwise} \end{array}\right..\nonumber
\end{align}
With this, the primal Lagrangian is given by $L_p(\boldsymbol{x},\delta_{12})=f_1(x_1)+f_2(x_2)+\delta_{12}(x_2-x_1)$, so that the dual function is given by $-f_1^{\ast}(\delta_{12})-f_2^{\ast}(-\delta_{12})$, where
\begin{align}
f_1^{\ast}(\delta_{12})=f_2^{\ast}(-\delta_{12})=\left\{\begin{array}{ll}\delta_{12} & 0\leq \delta_{12}\leq 1 \\ +\infty & \textrm{otherwise}\end{array}\right..\nonumber
\end{align}
Hence, the optimal solution for the primal and dual problem is $x_1^{\star}=x_2^{\star}\in [-1,1]$ and $\delta_{12}^{\star}=0$, respectively. The primal-dual Lagrangian in this case is given by
\begin{align}
L_{pd}(\boldsymbol{x},\boldsymbol{\lambda})=&f_1(x_1)+f_2(x_2)-f_1^{\ast}(\lambda_{1|2})-f_2^{\ast}(-\lambda_{2|1})\nonumber\\
&-x_1\lambda_{2|1}+x_2\lambda_{1|2}\label{equ:Lpd_example1}.
\end{align}
One can show that every point $(x_1',x_2',\lambda_{1|2}',\lambda_{2|1}')\in \{(x_1,x_2,0,0)|-1\leq x_1,x_2\leq 1\}$ is a saddle point of $L_{pd}(\boldsymbol{x},\boldsymbol{\lambda})$, which does not necessarily lead to $x_1'=x_2'$.
\end{example}


It is clear from Example~1 that finding a saddle point of $L_{pd}$ does not necessarily solve the primal problem (\ref{equ:optProMulti_re}). Similarly, one can also build another example illustrating that a saddle point of $L_{pd}$ does not necessarily solve the dual problem (\ref{equ:optProMultiDual2}). 

%
%

\vspace{-2mm}
\subsection{Augmented primal-dual Lagrangian}
\vspace{-1mm}

The problem that not every saddle point of $L_{pd}(\boldsymbol{x},\boldsymbol{\lambda})$ leads to an optimal point of the primal or dual problem can be solved by adding two quadratic penalty terms to $L_{pd}(\boldsymbol{x},\boldsymbol{\lambda})$ as
\begin{align}
\hspace{-1mm}L_{\mathcal{P}}(\boldsymbol{x},\boldsymbol{\lambda})=& L_{pd}(\boldsymbol{x},\boldsymbol{\lambda})+h_{\mathcal{P}_p}(\boldsymbol{x})-h_{\mathcal{P}_d}(\boldsymbol{\lambda}),
\label{equ:augPDLag}
\end{align}
where $h_{\mathcal{P}_p}(\boldsymbol{x})$ and $h_{\mathcal{P}_d}(\boldsymbol{\lambda})$ are defined as
\begin{align}
\hspace{-1mm}h_{\mathcal{P}_p}(\boldsymbol{x})=&\hspace{-1mm}\sum_{(i,j)\in \mathcal{E}} \hspace{-0.5mm}\frac{1}{2}\left\|\boldsymbol{A}_{i j}\boldsymbol{x}_{i}+\boldsymbol{A}_{j i}\boldsymbol{x}_{j}-\boldsymbol{c}_{ij}\right\|_{\boldsymbol{P}_{p,ij}}^2\label{equ:quadFunP} \\
h_{\mathcal{P}_d}(\boldsymbol{\lambda})=&\sum_{(i,j)\in \mathcal{E}}\frac{1}{2}\left\|\boldsymbol{\lambda}_{i|j}-\boldsymbol{\lambda}_{j|i}\right\|_{\boldsymbol{P}_{d,ij}}^2\hspace{-1mm},
\label{equ:quadFunD}
\end{align}
where $\mathcal{P}=\mathcal{P}_p\cup\mathcal{P}_d$ and
\begin{align}
\mathcal{P}_p&=\{\boldsymbol{P}_{p,ij}^T=\boldsymbol{P}_{p,ij}\succ 0|(i,j)\in \mathcal{E}\}\nonumber \\
\mathcal{P}_d&=\{\boldsymbol{P}_{d,ij}^T=\boldsymbol{P}_{d,ij}\succ 0 | (i,j)\in \mathcal{E}\}.\nonumber 
\end{align}
The set $\mathcal{P}$ of $2|\mathcal{E}|$ positive definite matrices remains to be specified. 

Let ${X}=\{\boldsymbol{x}| \boldsymbol{A}_{i j}\boldsymbol{x}_i+\boldsymbol{A}_{j i}\boldsymbol{x}_j=\boldsymbol{c}_{ij}, \forall (i,j)\in \mathcal{E}\}$ and ${\Lambda}=\{\boldsymbol{\lambda}|\boldsymbol{\lambda}_{i|j}=\boldsymbol{\lambda}_{j|i},\forall (i,j)\in \mathcal{E}\}$ denote the primal and dual feasible set, respectively. It is clear that $h_{\mathcal{P}_p}(\boldsymbol{x})\geq 0$ (or $-h_{\mathcal{P}_d}(\boldsymbol{\lambda})\leq 0$ ) with equality if and only if $\boldsymbol{x}\in X$ (or $\boldsymbol{\lambda}\in \Lambda$). The introduction of the two penalty functions essentially prevents non-feasible $\boldsymbol{x}$ and/or $\boldsymbol{\lambda}$ to correspond to saddle points of $L_{\mathcal{P}}(\boldsymbol{x},\boldsymbol{\lambda})$. As a consequence, we have a saddle point theorem for $L_{\mathcal{P}}$ which states that $\boldsymbol{x}^{\star}$ solves the primal problem (\ref{equ:optProMulti_re}) if and only if there exits $\boldsymbol{\lambda}^{\star}$ such that $(\boldsymbol{x}^{\star},\boldsymbol{\lambda}^{\star})$ is a saddle point of $L_{\mathcal{P}}(\boldsymbol{x},\boldsymbol{\lambda})$. To prove this result, we need the following lemma.



\begin{lemma}
Let $(\boldsymbol{x}^{\star},\boldsymbol{\lambda}^{\star})$ and $(\boldsymbol{x}',\boldsymbol{\lambda}')$ be two saddle points of $L_{\mathcal{P}}(\boldsymbol{x},\boldsymbol{\lambda})$. Then
\begin{align}
L_{\mathcal{P}}(\boldsymbol{x}',\boldsymbol{\lambda}^{\star})
\hspace{-0.3mm}=\hspace{-0.3mm}L_{\mathcal{P}}(\boldsymbol{x}',\boldsymbol{\lambda}')
\hspace{-0.3mm}=\hspace{-0.3mm}L_{\mathcal{P}}(\boldsymbol{x}^{\star},\boldsymbol{\lambda}^{\star})
\hspace{-0.3mm}=\hspace{-0.3mm}L_{\mathcal{P}}(\boldsymbol{x}^{\star},\boldsymbol{\lambda}').
\label{equ:saddlePointSet}
\end{align}
Further, $(\boldsymbol{x}',\boldsymbol{\lambda}^{\star})$ and $(\boldsymbol{x}^{\star},\boldsymbol{\lambda}')$ are two saddle points of $L_{\mathcal{P}}(\boldsymbol{x},\boldsymbol{\lambda})$ as well.
\label{lemma:saddlePointSet}
\end{lemma}
\begin{proof}
Since $(\boldsymbol{x}^{\star},\boldsymbol{\lambda}^{\star})$ and $(\boldsymbol{x}',\boldsymbol{\lambda}')$ are two saddle points of $L_{\mathcal{P}}(\boldsymbol{x},\boldsymbol{\lambda})$, we have
\begin{align}
L_{\mathcal{P}}(\boldsymbol{x}',\boldsymbol{\lambda}^{\star})
\leq & L_{\mathcal{P}}(\boldsymbol{x}',\boldsymbol{\lambda}')
\leq L_{\mathcal{P}}(\boldsymbol{x}^{\star},\boldsymbol{\lambda}') \nonumber \\
L_{\mathcal{P}}(\boldsymbol{x}^{\star},\boldsymbol{\lambda}')
\leq & L_{\mathcal{P}}(\boldsymbol{x}^{\star},\boldsymbol{\lambda}^{\star})
\leq L_{\mathcal{P}}(\boldsymbol{x}',\boldsymbol{\lambda}^{\star}).\nonumber
\end{align}
Combining the above two inequality chains produces (\ref{equ:saddlePointSet}). In order to show that $(\boldsymbol{x}',\boldsymbol{\lambda}^{\star})$ is a saddle point, we have $L_{\mathcal{P}}(\boldsymbol{x}',\boldsymbol{\lambda})
\leq L_{\mathcal{P}}(\boldsymbol{x}',\boldsymbol{\lambda}')=L_{\mathcal{P}}(\boldsymbol{x}',\boldsymbol{\lambda}^{\star})
=L_{\mathcal{P}}(\boldsymbol{x}^{\star},\boldsymbol{\lambda}^{\star})\leq L_{\mathcal{P}}(\boldsymbol{x},\boldsymbol{\lambda}^{\star})$. The proof for $(\boldsymbol{x}^{\star},\boldsymbol{\lambda}')$ is similar.  
\end{proof}

We are ready to prove the saddle point theorem for $L_{\mathcal{P}}(\boldsymbol{x},\boldsymbol{\lambda})$.

\begin{theorem} If $\boldsymbol{x}^{\star}$ solves the primal problem (\ref{equ:optProMulti_re}), there exists $\boldsymbol{\lambda}^{\star}$ such that $(\boldsymbol{x}^{\star},\boldsymbol{\lambda}^{\star})$ is a saddle point of $L_{\mathcal{P}}(\boldsymbol{x},\boldsymbol{\lambda})$. Conversely, if $(\boldsymbol{x}',\boldsymbol{\lambda}')$ is a saddle point of $L_{\mathcal{P}}(\boldsymbol{x},\boldsymbol{\lambda})$, then $\boldsymbol{x}'$ and $\boldsymbol{\lambda}'$ solves the primal and the dual problem, respectively. Or equivalently, the following optimality conditions hold
\begin{align}
\sum_{j\in \mathcal{N}_i}\hspace{-1.2mm}\boldsymbol{A}_{i j}^T\boldsymbol{\lambda}_{j|i}'\hspace{-0.4mm}\in \hspace{-0.4mm} \partial_{\boldsymbol{x}_i} f_i(\boldsymbol{x}_i') \hspace{2mm} & \forall i\hspace{-0.3mm}\in\hspace{-0.3mm} \mathcal{V} \label{equ:L_G_opt1} \\
\boldsymbol{A}_{i j}\boldsymbol{x}_i'+\boldsymbol{A}_{j i}\boldsymbol{x}_j'-\boldsymbol{c}_{ij}=\boldsymbol{0} \hspace{2mm}  & \forall (i,j)\hspace{-0.3mm}\in\hspace{-0.3mm} \mathcal{E}  \label{equ:L_G_opt3}\\
\boldsymbol{\lambda}_{i|j}'-\boldsymbol{\lambda}_{j|i}'=\boldsymbol{0} \hspace{2mm} & \forall (i,j)\hspace{-0.3mm}\in\hspace{-0.3mm} \mathcal{E}. \label{equ:L_G_opt4}
\end{align}
\label{thoerem:saddlePointGen}
\vspace{-5mm}
\end{theorem}
\begin{proof}
If $\boldsymbol{x}^{\star}$ solves the primal problem, then there exists a $\boldsymbol{\lambda}^{\star}$ such that $(\boldsymbol{x}^{\star},\boldsymbol{\lambda}^{\star})$ is a saddle point of $L_{pd}$ by Lemma~\ref{lemma:saddlePointForward}. Since $\boldsymbol{x}^{\star}\in X$ and $\boldsymbol{\lambda}^{\star}\in \Lambda$, we have $h_{\mathcal{P}_p}(\boldsymbol{x}^{\star})-h_{\mathcal{P}_d}(\boldsymbol{\lambda}^{\star})\hspace{-0.2mm}=0$, $\partial_{\boldsymbol{x}}h_{\mathcal{P}_p}(\boldsymbol{x}^{\star})\hspace{-0.2mm}=\hspace{-0.2mm}\boldsymbol{0}$ and $\partial_{\boldsymbol{\lambda}}h_{\mathcal{P}_d}(\boldsymbol{\lambda}^{\star})\hspace{-0.2mm}=\hspace{-0.2mm}\boldsymbol{0}$, from which we conclude that $(\boldsymbol{x}^{\star},\boldsymbol{\lambda}^{\star})$ is a saddle point of $L_{\mathcal{P}}(\boldsymbol{x},\boldsymbol{\lambda})$ as well.

Conversely, let $(\boldsymbol{x}',\boldsymbol{\lambda}')$ be a saddle point of $L_{\mathcal{P}}(\boldsymbol{x},\boldsymbol{\lambda})$. We first show that $\boldsymbol{x}'$ solves the primal problem. We have from Lemma~\ref{lemma:saddlePointSet} that $L_{\mathcal{P}}(\boldsymbol{x}',\boldsymbol{\lambda}^{\star})=L_{\mathcal{P}}(\boldsymbol{x}^{\star},\boldsymbol{\lambda}^{\star})$, which can be simplified as 
\begin{align}
&L_p(\boldsymbol{x}',\boldsymbol{\delta}^{\star})+L_d(\boldsymbol{\delta}^{\star},\boldsymbol{\lambda}^{\star},\boldsymbol{x}')+h_{\mathcal{P}_p}(\boldsymbol{x}')\nonumber \\
&=L_p(\boldsymbol{x}^{\star},\boldsymbol{\delta}^{\star})+L_d(\boldsymbol{\delta}^{\star},\boldsymbol{\lambda}^{\star},\boldsymbol{x}^{\star}),\nonumber \end{align}
from which we conclude that $h_{\mathcal{P}_p}(\boldsymbol{x}')=L_p(\boldsymbol{x}^{\star},\boldsymbol{\delta}^{\star})-L_p(\boldsymbol{x}',\boldsymbol{\delta}^{\star})\leq 0$ and thus $h_{\mathcal{P}_p}(\boldsymbol{x}')=0$ so that $\boldsymbol{x}'\in X$. In addition, since $(\boldsymbol{x}',\boldsymbol{\lambda}^{\star})$ is a saddle point of $L_{\mathcal{P}}(\boldsymbol{x},\boldsymbol{\lambda})$ by Lemma~\ref{lemma:saddlePointSet}, we have
\begin{align}
\sum_{j\in \mathcal{N}_i}\boldsymbol{A}_{i j}^T\boldsymbol{\delta}_{ij}^{\star}=\sum_{j\in \mathcal{N}_i}\boldsymbol{A}_{i j}^T\boldsymbol{\lambda}_{j|i}^{\star}\in \partial_{\boldsymbol{x}_i}f_i(\boldsymbol{x}_i'), \forall i\in \mathcal{V},\nonumber
\end{align}
and we conclude that $\boldsymbol{x}'$ solves the primal problem as required. Similarly, one can show that $\boldsymbol{\lambda}'$ solves the dual problem.

Based on the above analysis, we conclude that the optimality conditions for $(\boldsymbol{x}',\boldsymbol{\lambda}')$ being a saddle point of $L_{\mathcal{P}}$ are given by (\ref{equ:L_G_opt1})-(\ref{equ:L_G_opt4}). 
The set of optimality conditions 
$\{\boldsymbol{c}_{ij}\hspace{-0.5mm}-\hspace{-0.5mm}\boldsymbol{A}_{j i}\boldsymbol{x}_j' \hspace{-0.4mm}\in \hspace{-0.4mm} \partial_{\boldsymbol{\lambda}_{i|j}}\left[ f_i^{\ast}(\boldsymbol{A}_{i}^T \boldsymbol{\lambda}_{i}')\right] |[i,j] \in\hspace{-0.3mm} \vec{\mathcal{E}} \} $ is redundant and can be derived from (\ref{equ:L_G_opt1})-(\ref{equ:L_G_opt4}) (see (\ref{equ:conj_def})-(\ref{equ:conj_def3}) for the argument). 
\end{proof}

Theorem~\ref{thoerem:saddlePointGen} states that instead of solving the primal problem (\ref{equ:optProMulti_re}) or the dual problem (\ref{equ:optProMultiDual2}), one can alternatively search for a saddle point of $L_{\mathcal{P}}(\boldsymbol{x},\boldsymbol{\lambda})$. To briefly summarize, we consider solving the following min-max problem in the rest of the paper
\begin{align}
(\boldsymbol{x}^{\star},\boldsymbol{\lambda}^{\star})=\arg\min_{\boldsymbol{x}}\max_{\boldsymbol{\lambda}}L_{\mathcal{P}}(\boldsymbol{x},\boldsymbol{\lambda}). \label{equ:minMaxGen}
\end{align}
We will explain in next section how to iteratively approach the saddle point $(\boldsymbol{x}^{\star},\boldsymbol{\lambda}^{\star})$ in a distributed manner.

\section{Primal-Dual Method of Multipliers}
\label{sec:BiADMM}
In this section, we present a new algorithm named \emph{primal-dual method of multipliers} (PDMM) to iteratively approach a saddle point of $L_{\mathcal{P}}(\boldsymbol{x},\boldsymbol{\lambda})$. We propose both the synchronous and asynchronous PDMM for solving the problem. 

\vspace{-1.5mm}
\subsection{Synchronous updating scheme}
\label{subsec:syn_updating}
The synchronous updating scheme refers to the operation that at each iteration, all the variables over the graph update their estimates by using the most recent estimates from their neighbors from last iteration. Suppose $(\hat{\boldsymbol{x}}^{k},\hat{\boldsymbol{\lambda}}^{k})$ is the estimate obtained from the $k-1$th iteration, where $k\geq 1$. We compute the new estimate  $(\hat{\boldsymbol{x}}^{k+1},\hat{\boldsymbol{\lambda}}^{k+1})$ at iteration $k$ as
\begin{align}
\hspace{-2.5mm}\left(\hat{\boldsymbol{x}}_i^{k+1},\hat{\boldsymbol{\lambda}}_i^{k+1}\right)&\hspace{-1mm}=\hspace{-0.5mm}\arg\min_{\boldsymbol{x}_i}\max_{\boldsymbol{\lambda}_{ i}}\hspace{-0.5mm}L_{\mathcal{P}}\hspace{-0.5mm}\Big(\hspace{-0.6mm}\left[\ldots\hspace{-0.2mm},\hspace{-0.2mm}\hat{\boldsymbol{x}}_{i-1}^{k,T}\hspace{-0.2mm},\hspace{-0.2mm}\boldsymbol{x}_i^T\hspace{-0.2mm},\hspace{-0.2mm}\hat{\boldsymbol{x}}_{i+1}^{k,T}\hspace{-0.2mm},\hspace{-0.2mm}\ldots \right]^T\hspace{-.5mm},\nonumber\\
&\hspace{6mm}\left[\ldots\hspace{-0.2mm},\hspace{-0.2mm}\hat{\boldsymbol{\lambda}}_{i-1}^{k,T}\hspace{-0.2mm},\hspace{-0.2mm}\boldsymbol{\lambda}_i^T\hspace{-0.2mm},\hspace{-0.2mm}\hat{\boldsymbol{\lambda}}_{i+1}^{k,T}\hspace{-0.2mm},\hspace{-0.2mm}\ldots \right]^T\hspace{-1mm}\Big)\textrm{ }\textrm{ } i\in \mathcal{V}.\label{equ:x_lambda_updateSyn}
\end{align}
By inserting the expression (\ref{equ:augPDLag}) for $L_{\mathcal{P}}(\boldsymbol{x},\boldsymbol{\lambda})$ into (\ref{equ:x_lambda_updateSyn}), the updating expression can be further simplified as
\begin{align}
\hspace{-2mm}\hat{\boldsymbol{x}}_i^{k+1}\hspace{-1.5mm}=&\arg\min_{\boldsymbol{x}_i}\hspace{-1mm}\Bigg[
\hspace{-0.8mm}\sum_{j\in \mathcal{N}_i}\hspace{-0.8mm}\frac{1}{2}\left\|\boldsymbol{A}_{i j}\boldsymbol{x}_i+\boldsymbol{A}_{j i}\hat{\boldsymbol{x}}_j^{k}-\boldsymbol{c}_{ij}\right\|_{\boldsymbol{P}_{p,ij}}^2 \nonumber\\
&\hspace{10mm}-\hspace{-0.4mm}\boldsymbol{x}_i^{T}\Bigg(\hspace{-0.5mm}\sum_{j\in \mathcal{N}_i}\hspace{-1.5mm}\boldsymbol{A}_{i j}^T\hat{\boldsymbol{\lambda}}_{j|i}^{k}\hspace{-0.5mm}\Bigg)\hspace{-0.5mm}+\hspace{-0.5mm}f_i(\boldsymbol{x}_i)\hspace{-0.5mm}\Bigg]  \;\;\; i\in \mathcal{V}\label{equ:x_updateSyn} \\
\hspace{-2mm}\hat{\boldsymbol{\lambda}}_i^{k+1}\hspace{-1.5mm}=&\arg\min_{\boldsymbol{\lambda}_i}\hspace{-1mm} \Bigg[\hspace{-0.5mm}\sum_{j\in \mathcal{N}_i}\hspace{-1.5mm}\Bigg(\frac{1}{2}\left\|\boldsymbol{\lambda}_{i|j}-\hat{\boldsymbol{\lambda}}_{j|i}^{k}\right\|_{\boldsymbol{P}_{d,ij}}^2\hspace{-0.8mm}+\hspace{-0.4mm}\boldsymbol{\lambda}_{i|j}^{T}\boldsymbol{A}_{j i}\hat{\boldsymbol{x}}_{j}^{k} \nonumber \\
&\hspace{23mm}-\hspace{-0.5mm}\boldsymbol{\lambda}_{i|j}^T\boldsymbol{c}_{ij}\hspace{-0.3mm}\Bigg)\hspace{-0.5mm}
+\hspace{-0.5mm}f_i^{\ast}(\boldsymbol{A}_{i}^T\boldsymbol{\lambda}_{i})\hspace{-0.5mm}\Bigg]\;\;\; i\in \mathcal{V}.\label{equ:lambda_updateSyn}
\end{align}
Eq.~(\ref{equ:x_updateSyn})-(\ref{equ:lambda_updateSyn}) suggest that at iteration $k$, every node $i$ performs parameter-updating independently once the estimates $\{\hat{\boldsymbol{x}}_j^k,\hat{\boldsymbol{\lambda}}_{j|i}^k|j\in \mathcal{N}_i\}$ of its neighboring variables are available. In addition, the computation of $\hat{\boldsymbol{x}}_i^{k+1}$ and $\hat{\boldsymbol{\lambda}}_{i}^{k+1}$ can be carried out in parallel since $\boldsymbol{x}_i$ and $\boldsymbol{\lambda}_i$ are not directly related in $L_{\mathcal{P}}(\boldsymbol{x},\boldsymbol{\lambda})$. We refer to (\ref{equ:x_updateSyn})-(\ref{equ:lambda_updateSyn}) as \emph{node-oriented} computation. 

In order to run PDMM over the graph, each iteration should consist of two steps. Firstly, every node $i$ computes $(\hat{\boldsymbol{x}}_{i},\hat{\boldsymbol{\lambda}}_i)$ by following (\ref{equ:x_updateSyn})-(\ref{equ:lambda_updateSyn}), accounting for \emph{information-fusion}. Secondly, every node $i$ sends $(\hat{\boldsymbol{x}}_i,\hat{\boldsymbol{\lambda}}_{i|j})$ to its neighboring node $j$ for all neighbors, accounting for \emph{information-spread}. 
We take $\hat{\boldsymbol{x}}_i$ as the common message to all neighbors of node $i$ and $\hat{\boldsymbol{\lambda}}_{i|j}$ as a node-specific message only to neighbor $j$. In some applications, it may be preferable to exploit broadcast transmission rather than point-to-point transmission in order to save energy.  We will explain in Subsection~\ref{subsec:node_comp_trans} that the transmission of $\hat{\boldsymbol{\lambda}}_{i|j}$, $j\in \mathcal{N}_i$, can be replaced by broadcast transmission of an intermediate quantity.

Finally, we consider terminating the iterates (\ref{equ:x_updateSyn})-(\ref{equ:lambda_updateSyn}). One can check if the estimate $(\hat{\boldsymbol{x}},\hat{\boldsymbol{\lambda}})$ becomes stable over consecutive iterates (see Corollary~\ref{coro:syn} for theoretical support). 


\vspace{-1.5mm}
\subsection{Asynchronous updating scheme}
\label{subsec:asyn_updating}
The asynchronous updating scheme refers to the operation that at each iteration, only the variables associated with one node in the graph update their estimates while all other variables keep their estimates fixed. Suppose node $i$ is selected at iteration $k$. We then compute $(\hat{\boldsymbol{x}}_i^{k+1},\hat{\boldsymbol{\lambda}}_i^{k+1})$ by optimizing $L_{\mathcal{P}}$ based on the most recent estimates $\{\hat{\boldsymbol{x}}_j^{k},\hat{\boldsymbol{\lambda}}_{j|i}^{k}|j\in \mathcal{N}_i\}$ from its neighboring nodes. At the same time, the estimates $(\hat{\boldsymbol{x}}_j^k,\hat{\boldsymbol{\lambda}}_j^{k})$, $j\neq i$, remain the same. By following the above computational instruction, $(\hat{\boldsymbol{x}}^{k+1},\hat{\boldsymbol{\lambda}}^{k+1})$ can be obtained as
\begin{align}
&\hspace{-2mm}\left(\hspace{-0.3mm}\hat{\boldsymbol{x}}_i^{k+1},\hspace{-0.2mm}\hat{\boldsymbol{\lambda}}_i^{k+1}\hspace{-0.3mm}\right)\hspace{-0.7mm}=\hspace{-0.7mm}\arg\min_{\boldsymbol{x}_i}\max_{\boldsymbol{\lambda}_{ i}}\hspace{-0.5mm}L_{\mathcal{P}}\hspace{-0.2mm}\Big(\hspace{-0.6mm}\left[\ldots\hspace{-0.2mm},\hspace{-0.2mm}\hat{\boldsymbol{x}}_{i-1}^{k,T}\hspace{-0.2mm},\hspace{-0.2mm}\boldsymbol{x}_i^T\hspace{-0.2mm},\hspace{-0.2mm}\hat{\boldsymbol{x}}_{i+1}^{k,T}\hspace{-0.2mm},\hspace{-0.2mm}\ldots \right]^T\hspace{-.6mm},\nonumber\\
&\hspace{38mm}\left[\ldots\hspace{-0.2mm},\hspace{-0.2mm}\hat{\boldsymbol{\lambda}}_{i-1}^{k,T}\hspace{-0.2mm},\hspace{-0.2mm}\boldsymbol{\lambda}_i^T\hspace{-0.2mm},\hspace{-0.2mm}\hat{\boldsymbol{\lambda}}_{i+1}^{k,T}\hspace{-0.2mm},\hspace{-0.2mm}\ldots \right]^T\hspace{-1mm}\Big)\label{equ:x_lambda_updateAsyn1} \\
&(\hat{\boldsymbol{x}}_{j}^{k+1},\hat{\boldsymbol{\lambda}}_{j}^{k+1})=(\hat{\boldsymbol{x}}_{j}^{k},\hat{\boldsymbol{\lambda}}_{j}^{k}) \hspace{20mm} j\in \mathcal{V}, j\neq i.  \label{equ:x_lambda_updateAsyn2}
\end{align}

Similarly to (\ref{equ:x_updateSyn})-(\ref{equ:lambda_updateSyn}), $\hat{\boldsymbol{x}}_i^{k+1}$ and $\hat{\boldsymbol{\lambda}}_i^{k+1}$ can also be computed separately in (\ref{equ:x_lambda_updateAsyn1}). Once the update at node $i$ is complete, the node sends the common message $\hat{\boldsymbol{x}}_i^{k+1}$ and node-specific messages $\{\hat{\boldsymbol{\lambda}}_{i|j}^{k+1},j\in \mathcal{N}_i\}$ to its neighbors 
We will explain in next subsection how to exploit broadcast transmission to replace point-to-point transmission.

In practice, the nodes in the graph can either be randomly activated or follow a predefined order for asynchronous parameter-updating. One scheme for realizing random node-activation is that after a node finishes parameter-updating, it randomly activates one of its neighbors for next iteration. Another scheme is to introduce a clock at each node which ticks at the times of a (random) Poisson process (see \cite{Boyd06gossip} for detailed information). Each node is activated only when its clock ticks. As for node-activation in a predefined order, cyclic updating scheme is probably most straightforward. Once node $i$ finishes parameter-updating, it informs node $i+1$ for next iteration. For the case that node $i$ and $i+1$ are not neighbors, the path from node $i$ to $i+1$ can be pre-stored at node $i$ to facilitate the process. In Subsection~\ref{subsec:rate_asyn}, we provide convergence analysis only for the cyclic updating scheme. We leave the analysis for other asynchronous schemes for future investigation.   


\begin{remark}
To briefly summarize, synchronous PDMM scheme allows faster information-spread over the graph through parallel parameter-updating while asynchronous PDMM scheme requires less effort from node-coordination in the graph. In practice, the scheme-selection should depend on the graph (or network) properties such as the feasibility of parallel computation, the complexity of node-coordination and the life time of nodes.      
\end{remark}

\vspace{-1.5mm}
\subsection{Simplifying node-based computations and transmissions}
\label{subsec:node_comp_trans}

It is clear that for both the synchronous and asynchronous schemes, each activated node $i$ has to perform two minimizations: one for $\hat{\boldsymbol{x}}_i$ and the other one for $\hat{\boldsymbol{\lambda}}_i$. In this subsection, we show that the computations for the two minimizations can be simplified. We will also study how the point-to-point transmission can be replaced with broadcast transmission. To do so, we will consider two scenarios: 


\subsubsection{Avoiding conjugate functions} \label{subsub:avoid_conj} In the first scenario, we consider using $f_i(\cdot)$ instead of $f_i^{\ast}(\cdot)$ to update ${\hat{\boldsymbol{\lambda}}}_i$. Our goal is to simplify computations by avoiding the derivation of $f_i^{\ast}(\cdot)$.

By using the definition of $f_i^{\ast}$ in (\ref{equ:conj_def}), the computation (\ref{equ:lambda_updateSyn}) for $\hat{\boldsymbol{\lambda}}_i^{k+1}$ (which also holds for asynchronous PDMM) can be rewritten as
\begin{align}
\hspace{-2mm}\hat{\boldsymbol{\lambda}}_i^{k+1}\hspace{-1.5mm}=&\arg\min_{\boldsymbol{\lambda}_i}\hspace{-1mm} \Bigg[\hspace{-0.5mm}\sum_{j\in \mathcal{N}_i}\hspace{-1.5mm}\Bigg(\frac{1}{2}\left\|\boldsymbol{\lambda}_{i|j}-\hat{\boldsymbol{\lambda}}_{j|i}^{k}\right\|_{\boldsymbol{P}_{d,ij}}^2\hspace{-0.8mm}+\hspace{-0.4mm}\boldsymbol{\lambda}_{i|j}^{T}\boldsymbol{A}_{j i}\hat{\boldsymbol{x}}_{j}^{k} \nonumber \\
&\hspace{11mm}-\hspace{-0.5mm}\boldsymbol{\lambda}_{i|j}^T\boldsymbol{c}_{ij}\hspace{-0.3mm}\Bigg)\hspace{-0.5mm}
+\hspace{-0.5mm}\max_{\boldsymbol{w}_i}\hspace{-0.5mm}\Big(\boldsymbol{w}_i^T\boldsymbol{A}_{i}^T\boldsymbol{\lambda}_{i}\hspace{-0.5mm}-\hspace{-0.5mm}f_i(\boldsymbol{w}_i)\Big)\hspace{-0.5mm}\Bigg].\label{equ:lambda_update}
\end{align}
We denote the optimal solution for $\boldsymbol{w}_i$ in (\ref{equ:lambda_update}) as $\boldsymbol{w}_i^{k+1}$. The optimality conditions for $\hat{\boldsymbol{\lambda}}_{i|j}^{k+1}$, $j\in \mathcal{N}_i$, and $\boldsymbol{w}_i^{k+1}$ can then be derived from (\ref{equ:lambda_update}) as
\begin{align}
&\hspace{-2mm}\boldsymbol{0}\in\boldsymbol{A}_i^T\hat{\boldsymbol{\lambda}}_i^{k+1}- \partial_{\boldsymbol{w}_i} f_i(\boldsymbol{w}_i^{k+1}) \label{equ:w_lambda_optCond1} \\
&\hspace{-2mm}\boldsymbol{c}_{ij}\hspace{-0.5mm}=\hspace{-0.5mm}\boldsymbol{P}_{d,ij}(\hat{\boldsymbol{\lambda}}_{i|j}^{k+1}\hspace{-0.5mm}-\hspace{-0.5mm}\hat{\boldsymbol{\lambda}}_{j|i}^k)\hspace{-0.5mm}+\hspace{-0.5mm}\boldsymbol{A}_{ji}\hat{\boldsymbol{x}}_j^k\hspace{-0.5mm}+\hspace{-0.5mm}\boldsymbol{A}_{ij}\boldsymbol{w}_i^{k+1} \hspace{1mm}j\in \mathcal{N}_i, \label{equ:w_lambda_optCond2}
\end{align}
where (\ref{equ:A_lambda_relation}) is used in deriving (\ref{equ:w_lambda_optCond2}). Since $\boldsymbol{P}_{d,ij}$ is a nonsingular matrix, (\ref{equ:w_lambda_optCond2}) defines a mapping from $\boldsymbol{w}_i^{k+1}$ to $\hat{\boldsymbol{\lambda}}_{i|j}^{k+1}$: 
\begin{align}
\hspace{-1mm}\hat{\boldsymbol{\lambda}}_{i|j}^{k+1}&\hspace{-0.7mm}=\hspace{-0.5mm}\hat{\boldsymbol{\lambda}}_{j|i}^k\hspace{-0.5mm}+\hspace{-0.5mm}\boldsymbol{P}_{d,ij}^{-1}(\hspace{-0.3mm}\boldsymbol{c}_{ij}\hspace{-0.5mm}-\hspace{-0.5mm}\boldsymbol{A}_{ji}\hat{\boldsymbol{x}}_j^k\hspace{-0.5mm}-\hspace{-0.5mm}\boldsymbol{A}_{ij}\boldsymbol{w}_i^{k+1}\hspace{-0.5mm}), j\in \hspace{-0.4mm} \mathcal{N}_i, \label{equ:lambda_x_relation1}
\end{align}
With this mapping, (\ref{equ:w_lambda_optCond1}) can then be reformulated as 
\begin{align}
&\sum_{j\in\mathcal{N}_i}\hspace{-1.5mm}\boldsymbol{A}_{ij}^T\left(\hat{\boldsymbol{\lambda}}_{j|i}^k\hspace{-0.5mm}+\hspace{-0.5mm}\boldsymbol{P}_{d,ij}^{-1}(\hspace{-0.3mm}\boldsymbol{c}_{ij}\hspace{-0.5mm}-\hspace{-0.5mm}\boldsymbol{A}_{ji}\hat{\boldsymbol{x}}_j^k\hspace{-0.5mm}-\hspace{-0.5mm}\boldsymbol{A}_{ij}\boldsymbol{w}_i^{k+1}\hspace{-0.5mm})\right)
\nonumber\\
&\in \partial_{\boldsymbol{w}_i} f_i(\boldsymbol{w}_i^{k+1}).  \label{equ:w_lambda_optCond4}
\end{align}
By inspection of (\ref{equ:w_lambda_optCond4}), it can be shown that (\ref{equ:w_lambda_optCond4}) is in fact an optimality condition for the following optimization problem
\begin{align}
\boldsymbol{w}_i^{k+1}\hspace{-0.2mm}=\hspace{-0.2mm}\arg\min_{\boldsymbol{w}_i} & \Big[f_i(\boldsymbol{w}_i)+\frac{1}{2}\|\boldsymbol{c}_{ij}\hspace{-0.5mm}-\hspace{-0.5mm}\boldsymbol{A}_{ji}\hat{\boldsymbol{x}}_j^k-\boldsymbol{A}_{ij}\boldsymbol{w}_i \|_{\boldsymbol{P}_{d,ij}^{-1}}^2\nonumber\\
&\hspace{1.5mm}-\boldsymbol{w}_i^T\hspace{-1.mm}\sum_{j\in\mathcal{N}_i}\hspace{-1.5mm}\boldsymbol{A}_{ij}^T\hat{\boldsymbol{\lambda}}_{j|i}^k\hspace{-0.5mm}\Big].
\label{equ:w_lambda_optCond5}
\end{align}

The above analysis suggests that $\hat{\boldsymbol{\lambda}}_i^{k+1}$ can be alternatively computed through an intermediate quantity $\boldsymbol{w}_i^{k+1}$. We summarize the result in a proposition below.

\begin{proposition}
Considering a node $i\in \mathcal{V}$ at iteration $k$, the new estimate $\hat{\boldsymbol{\lambda}}_{i|j}^{k+1}$ for each $j\in \mathcal{N}_i$ can be obtained by following (\ref{equ:lambda_x_relation1}),
where ${\boldsymbol{w}}_i^{k+1}$ is computed by (\ref{equ:w_lambda_optCond5}).
\label{prop:no_conjugate}
\end{proposition}

Proposition~\ref{prop:no_conjugate} suggests that the estimate $\hat{\boldsymbol{\lambda}}_i^{k+1}$ can be easily computed from ${\boldsymbol{w}}_i^{k+1}$. We argue in the following that the point-to-point transmission of $\left\{\hat{\boldsymbol{\lambda}}_{i|j}^{k+1},j\in \mathcal{N}_i\right\}$ can be replaced with broadcast transmission of ${\boldsymbol{w}}_i^{k+1}$.


We see from (\ref{equ:lambda_x_relation1}) that the computation of the node-specific message $\hat{\boldsymbol{\lambda}}_{i|j}^{k+1}$ (from node $i$ to node $j$) only consists of the quantities ${\boldsymbol{w}}_i^{k+1}$, $\hat{\boldsymbol{\lambda}}_{j|i}^k$ and $\hat{\boldsymbol{x}}_{j}^k$. Since $\hat{\boldsymbol{\lambda}}_{j|i}^k$ and $\hat{\boldsymbol{x}}_{j}^k$ are available at node $j$, the message $\hat{\boldsymbol{\lambda}}_{i|j}^{k+1}$ can therefore be computed at node $j$ once the common message ${\boldsymbol{w}}_i^{k+1}$ is received. In other words, it is sufficient for node $i$ to broadcast both $\hat{\boldsymbol{x}}_i^{k+1}$ and ${\boldsymbol{w}}_i^{k+1}$ to all its neighbors. Every node-specific message $\hat{\boldsymbol{\lambda}}_{i|j}^{k+1}$, $j\in \mathcal{N}_i$, can then be computed at node $j$ alone. 

Finally, in order for the broadcast transmission to work, we assume there is no transmission failure between neighboring nodes. The assumption ensures that there is no estimate inconsistency between neighboring nodes, making the broadcast transmission reliable.

\subsubsection{Reducing two minimizations to one} In the second scenario, we study under what conditions the two minimizations (\ref{equ:x_updateSyn})-(\ref{equ:lambda_updateSyn}) (which also hold for asynchronous PDMM) reduce to one minimization. 

\begin{proposition}
Considering a node $i\in \mathcal{V}$ at iteration $k$, if the matrix $\boldsymbol{P}_{d,ij}$ for every neighbor $j\in \mathcal{N}_i$ is chosen to be $\boldsymbol{P}_{d,ij}=\boldsymbol{P}_{p,ij}^{-1}$, then there is $\hat{\boldsymbol{x}}_i^{k+1}={\boldsymbol{w}}_i^{k+1}$. As a result,
\begin{align}
\hspace{-1mm}\hat{\boldsymbol{\lambda}}_{i|j}^{k+1}\hspace{-0.5mm}=\hspace{-0.5mm}\hat{\boldsymbol{\lambda}}_{j|i}^k\hspace{-0.5mm}+\hspace{-0.5mm}\boldsymbol{P}_{p,ij}(\boldsymbol{c}_{ij}\hspace{-0.5mm}-\hspace{-0.5mm}\boldsymbol{A}_{j i}\hat{\boldsymbol{x}}_j^k\hspace{-0.5mm}-\hspace{-0.5mm}\boldsymbol{A}_{i j}\hat{\boldsymbol{x}}_i^{k+1})\textrm{ } j\in \mathcal{N}_i. \label{equ:lambda_x_relation2}
\end{align}
\label{prop:lambda_x_relation}
\end{proposition}
\vspace{-1mm}
\begin{proof}
The proof is trivial. By inspection of (\ref{equ:x_updateSyn}) and (\ref{equ:w_lambda_optCond5}) under $\boldsymbol{P}_{d,ij}\hspace{-0.3mm}=\hspace{-0.3mm}\boldsymbol{P}_{p,ij}^{-1}$, $j\in\hspace{-0.3mm} \mathcal{N}_i$, we obtain $\hat{\boldsymbol{x}}_i^{k+1}={\boldsymbol{w}}_i^{k+1}$.
\end{proof}


\begin{table}[t!]
\centering
\begin{tabular}{l}
  \hline
   \textrm{Initialization}: Properly initialize $\{\boldsymbol{x}_i\}$ and $\{\boldsymbol{\lambda}_{i|j}\}$\\
   Repeat\\
   \hspace{3mm} for all $i\in \mathcal{V}$ do \\
   \hspace{5mm}  $\hat{\boldsymbol{x}}_i^{k+1}\hspace{-0.5mm}=\hspace{-0.5mm}\arg\min_{\boldsymbol{x}_i}\hspace{-1mm}\Big[f_i(\boldsymbol{x}_i)\hspace{-0.5mm}-\hspace{-0.5mm} \boldsymbol{x}_i^T(\sum_{j\in \mathcal{N}_i}\hspace{-0.5mm}\boldsymbol{A}_{ij}^T\hat{\boldsymbol{\lambda}}_{j|i}^k)$ \\
	\hspace{18mm}$+\sum_{j\in \mathcal{N}_i}\frac{1}{2}\| \boldsymbol{A}_{ij}\boldsymbol{x}_i\hspace{-0.5mm} +\hspace{-0.5mm}\boldsymbol{A}_{ji}\hat{\boldsymbol{x}}_j^{k}-\boldsymbol{c}_{ij}\|_{\boldsymbol{P}_{p,ij}}^2\Big]$ \\
	 \hspace{3mm} end for \\
   \hspace{3mm} for all $i\in \mathcal{V}$ and $j\in \mathcal{N}_i$ do \\
   \hspace{5mm} $\hat{\boldsymbol{\lambda}}_{i|j}^{k+1}=\hat{\boldsymbol{\lambda}}_{j|i}^{k}+\boldsymbol{P}_{p,ij}(\boldsymbol{c}_{ij}-\boldsymbol{A}_{ji}\hat{\boldsymbol{x}}_j^{k}-\boldsymbol{A}_{ij}\hat{\boldsymbol{x}}_i^{k+1})$ \\
	 \hspace{3mm} end for \\
	 \hspace{3mm}$k\leftarrow k+1$ \\
	Until some stopping criterion is met \\
  \hline
\end{tabular}
\caption{Synchronous PDMM where for each $i\in \mathcal{V}$, $\boldsymbol{P}_{d,ij}=\boldsymbol{P}_{p,ij}^{-1}$ . }
\label{table:PDMM}
\vspace{-8mm}
\end{table}

Similarly to the first scenario, broadcast transmission is also applicable for the second scenario. Since $\hat{\boldsymbol{x}}_i^{k+1}={\boldsymbol{w}}_i^{k+1}$, node $i$ only has to broadcast the estimate $\hat{\boldsymbol{x}}_i^{k+1}$ to all its neighbors. Each message $\hat{\boldsymbol{\lambda}}_{i|j}^{k+1}$ from node $i$ to node $j$ can then be computed at node $j$ directly by applying (\ref{equ:lambda_x_relation2}).  See Table~\ref{table:PDMM} for the procedure of synchronous PDMM.    

\vspace{-1mm}
\section{Convergence Analysis}
\label{sec:convAnalysis}
In this section, we analyze the convergence rates of PDMM for both the synchronous and asynchronous  schemes. 
Inspired by the convergence analysis of ADMM \cite{Wang12OADM,Deng16ADMM}, we construct a special inequality (presented in \ref{subsec:inequality}) for $L_{\mathcal{P}}(\boldsymbol{x},\boldsymbol{\lambda})$ and then exploit it to analyze both synchronous PDMM (presented in \ref{subsec:rate_syn}) and asynchronous PDMM  (presented in \ref{subsec:rate_asyn}).

Before constructing the inequality, we first study how to properly choose the matrices in the set $\mathcal{P}$ (presented in \ref{subsec:param_selection}) in order to enable convergence analysis. 

\subsection{Parameter setting}
\label{subsec:param_selection}

In order to analyze the algorithm convergence later on, we first have to select the matrix set $\mathcal{P}$ properly. We impose a condition on each pair of matrices $(\boldsymbol{P}_{p,ij}\succ \boldsymbol{0},\boldsymbol{P}_{d,ij}\succ\boldsymbol{0})$, $(i,j)\in \mathcal{E}$, in $L_{\mathcal{P}}$:

\begin{condition}
In the function $L_{\mathcal{P}}$, each matrix $\boldsymbol{P}_{d,ij}$ can be represented in terms of $\boldsymbol{P}_{p,ij}$ as
\begin{align}
\boldsymbol{P}_{d,ij}=\boldsymbol{P}_{p,ij}^{-1}+\Delta \boldsymbol{P}_{d,ij}\quad \forall (i,j)\in \mathcal{E}, \label{equ:G_dp}
\end{align}
where $\Delta \boldsymbol{P}_{d,ij}\succeq \boldsymbol{0}$.
\label{con:G}
\end{condition}

Eq.~(\ref{equ:G_dp}) implies that $\boldsymbol{P}_{p,ij}$ and $\boldsymbol{P}_{d,ij}$ can not be chosen arbitrarily for our convergence analysis. If $\boldsymbol{P}_{p,ij}$ is small, then $\boldsymbol{P}_{d,ij}$ has to be chosen big enough to make (\ref{equ:G_dp}) hold, and vice versa. One special setup for $(\boldsymbol{P}_{p,ij},\boldsymbol{P}_{d,ij})$ is to let $\boldsymbol{P}_{d,ij}=\boldsymbol{P}_{p,ij}^{-1}$, or equivalently, $\Delta \boldsymbol{P}_{d,ij}=\boldsymbol{0}$. This leads to the application of Proposition~\ref{prop:lambda_x_relation}, which reduces two minimizations to one minimization for each activated node.

One simple setup in Condition~\ref{con:G} is to let all the matrices in $\mathcal{P}$ take scalar form. That is setting $(\boldsymbol{P}_{p,ij},\boldsymbol{P}_{d,ij})$, $(i,j)\in \mathcal{E}$, to be identity matrices multiplied by positive parameters:
\begin{align}
(\boldsymbol{P}_{p,ij},\boldsymbol{P}_{d,ij})=(\gamma_{p,ij}\boldsymbol{I}_{n_{ij}}, \gamma_{d,ij} \boldsymbol{I}_{n_{ij}}) \label{equ:gamma_setup}
\end{align}
where $\gamma_{p,ij}>0$, $ \gamma_{d,ij}> 0$ and $\gamma_{d,ij}\gamma_{p,ij}\geq 1$. It is worth noting that matrix form of $(\boldsymbol{P}_{p,ij},\boldsymbol{P}_{d,ij})$ might lead to faster convergence for some optimization problems.



\subsection{Constructing an inequality}
\label{subsec:inequality}

Before introducing the inequality, we first define a new function which involves $\{f_i, i\in\mathcal{V}\}$ and their conjugates:  
\begin{align}
p(\boldsymbol{x},\boldsymbol{\lambda})\hspace{-0mm}=\sum_{i\in \mathcal{V}} \Big[f_i(\boldsymbol{x}_i)\hspace{-0.5mm}+\hspace{-0.5mm}f_i^{\ast}(\boldsymbol{A}_{i}^{T}\boldsymbol{\lambda}_{i})-\frac{1}{2}\hspace{-0.6mm}\sum_{j\in \mathcal{N}_i}\hspace{-0.7mm}\boldsymbol{c}_{ij}^T\boldsymbol{\lambda}_{i|j}\Big].
\label{equ:p_x_lambda}
\end{align}
By studying (\ref{equ:optProMulti_re}) and (\ref{equ:optProMultiDual2}) at a saddle point $(\boldsymbol{x}^{\star},\boldsymbol{\lambda}^{\star})$ of $L_{\mathcal{P}}$, one can show that $p(\boldsymbol{x}^{\star},\boldsymbol{\lambda}^{\star})=0$. 

With $p(\boldsymbol{x},\boldsymbol{\lambda})$, the inequality for $L_{\mathcal{P}}$ can be described as:

\begin{lemma}
Let $(\boldsymbol{x}^{\star},\boldsymbol{\lambda}^{\star})$ be a saddle point of $L_{\mathcal{P}}$. Then for any $(\boldsymbol{x},\boldsymbol{\lambda})$, there is
\begin{align}
0\leq& \sum_{i\in \mathcal{V}}\sum_{j\in \mathcal{N}_i}\Big[(\boldsymbol{\lambda}_{i|j}-\boldsymbol{\lambda}_{i|j}^{\star})^T\Big(\boldsymbol{A}_{j i}\boldsymbol{x}_j-\frac{\boldsymbol{c}_{ij}}{2}\Big)\nonumber\\
&\hspace{16mm}-(\boldsymbol{x}_i-\boldsymbol{x}_i^{\star})^T\boldsymbol{A}_{i j}^T\boldsymbol{\lambda}_{j|i}\Big] +p(\boldsymbol{x},\boldsymbol{\lambda}), \label{equ:VI_gen1}
\end{align}
where equality holds if and only if $(\boldsymbol{x},\boldsymbol{\lambda})$ satisfies
\begin{align}
&\hspace{-2mm}\boldsymbol{0}\in \partial_{\boldsymbol{x}_i}f_i(\boldsymbol{x}_i^{\star})\hspace{-0.5mm}-\hspace{-01mm}\sum_{j\in \mathcal{N}_i}\hspace{-1.2mm}\boldsymbol{A}_{i j}^T\boldsymbol{\lambda}_{j|i} & \forall i\in \mathcal{V}
\label{equ:p_opt_cond2} \\
&\hspace{-2mm}\boldsymbol{0}\in \partial_{\boldsymbol{x}_i}f_i(\boldsymbol{x}_i)\hspace{-0.5mm}-\hspace{-01mm}\sum_{j\in \mathcal{N}_i}\hspace{-1.2mm}\boldsymbol{A}_{i j}^T\boldsymbol{\lambda}_{j|i}^{\star} & \forall i\in \mathcal{V}.
\label{equ:p_opt_cond4}
\end{align}
\label{lemma:inequality}
\end{lemma}

\begin{proof}
Given a saddle point $(\boldsymbol{x}^{\star},\boldsymbol{\lambda}^{\star})$ of $L_{\mathcal{P}}$, the right hand side of the inequality (\ref{equ:VI_gen1}) can be reformulated as 
\begin{align}
&\hspace{-1mm} \sum_{i\in \mathcal{V}}\hspace{-0.6mm}\Bigg[\hspace{-0.5mm}\sum_{j\in \mathcal{N}_i}\hspace{-1.5mm}\Big(\hspace{-0.3mm} -\hspace{-0.3mm}\boldsymbol{\lambda}_{i|j}^{\star,T}\Big(\boldsymbol{A}_{j i}\boldsymbol{x}_j-\frac{\boldsymbol{c}_{ij}}{2}\Big)+\boldsymbol{x}_i^{\star,T}\boldsymbol{A}_{i j}^T\boldsymbol{\lambda}_{j|i} \nonumber \\
&\hspace{18mm}-\hspace{-0.5mm}\boldsymbol{\lambda}_{i|j}^T\boldsymbol{c}_{ij}\Big)\hspace{-0.5mm}+\hspace{-0.5mm}f_i(\boldsymbol{x}_i)\hspace{-0.5mm}+\hspace{-0.5mm}f_i^{\ast}(\boldsymbol{A}_{i}^T\boldsymbol{\lambda}_{i})\hspace{-0.4mm}\Bigg] \nonumber  \\
&\hspace{-2mm}=\hspace{-0.5mm} \sum_{i\in \mathcal{V}}\hspace{-0.6mm}\Bigg[\hspace{-0.5mm}\sum_{j\in \mathcal{N}_i}\hspace{-1.5mm}\Big(\hspace{-0.3mm}-\hspace{-0.3mm}\boldsymbol{\lambda}_{j|i}^{\star,T}\boldsymbol{A}_{i j}\boldsymbol{x}_i+(\boldsymbol{A}_{j i}\boldsymbol{x}_j^{\star}-\boldsymbol{c}_{ij})^T\boldsymbol{\lambda}_{i|j}\Big) \nonumber \\
&\hspace{8mm}+\hspace{-0.5mm}f_i(\boldsymbol{x}_i)\hspace{-0.5mm}+\hspace{-0.5mm}f_i^{\ast}(\boldsymbol{A}_{i}^T\boldsymbol{\lambda}_{i})+\frac{1}{2}\hspace{-0.6mm}\sum_{j\in \mathcal{N}_i}\hspace{-0.7mm}\boldsymbol{c}_{ij}^T\boldsymbol{\lambda}_{i|j}^{\star}\hspace{-0.4mm}\Bigg]
\nonumber \\
&\hspace{-2mm}=\hspace{-0.5mm} \sum_{i\in \mathcal{V}}\hspace{-0.6mm}\Bigg[\hspace{-0.5mm}\sum_{j\in \mathcal{N}_i}\hspace{-1.5mm}\Big(\hspace{-0.3mm}-\hspace{-0.3mm}\boldsymbol{\lambda}_{i|j}^{\star,T}\boldsymbol{A}_{i j}\boldsymbol{x}_i-\boldsymbol{x}_i^{\star,T}\boldsymbol{A}_{i j}^T\boldsymbol{\lambda}_{i|j}\Big)\hspace{-0.5mm}+\hspace{-0.5mm}f_i(\boldsymbol{x}_i)\nonumber \\
&\hspace{8mm}+\hspace{-0.5mm}f_i^{\ast}(\boldsymbol{A}_{i}^T\boldsymbol{\lambda}_{i})+\frac{1}{2}\hspace{-0.6mm}\sum_{j\in \mathcal{N}_i}\hspace{-0.7mm}\boldsymbol{c}_{ij}^T\boldsymbol{\lambda}_{i|j}^{\star}\hspace{-0.4mm}\Bigg],
\label{equ:VI_gen2}
\end{align}
where the last equality is obtained by using $(\boldsymbol{x}^{\star},\boldsymbol{\lambda}^{\star})\in (X,\Lambda)$. 
Using Fenchel's inequalities (\ref{equ:Fenchel_ineq}), we conclude that for any $i\in \mathcal{V}$, the following two inequalities hold
\begin{align}
\displaystyle f_i^{\ast}(\boldsymbol{A}_{i}^T\boldsymbol{\lambda}_{i})\hspace{-0.5mm}-\hspace{-0.5mm}\boldsymbol{x}_i^{\star,T}\hspace{-0.5mm}(\boldsymbol{A}_{i}^T\boldsymbol{\lambda}_{i})&\geq -\hspace{-0.5mm}f_i(\boldsymbol{x}_i^{\star}) \label{equ:VI_gen3} \\
\displaystyle f_i(\boldsymbol{x}_i)-\boldsymbol{x}_i^{T}(\boldsymbol{A}_{i}^T\boldsymbol{\lambda}_{i})^{\star} &\geq \hspace{-0.5mm} -f_i^{\ast}(\boldsymbol{A}_{i}^T\boldsymbol{\lambda}_{i}^{\star}). \label{equ:VI_gen4}
\end{align}

Finally, combining (\ref{equ:VI_gen2})-(\ref{equ:VI_gen4}) and the fact that $p(\boldsymbol{x}^{\star},\boldsymbol{\lambda}^{\star})=0$ produces the inequality (\ref{equ:VI_gen1}).
The equality holds if and only if (\ref{equ:VI_gen3})-(\ref{equ:VI_gen4}) hold, of which the optimality conditions are given by (\ref{equ:p_opt_cond2})-(\ref{equ:p_opt_cond4}) (see (\ref{equ:conj_def})-(\ref{equ:conj_def3}) for the argument).
\end{proof}

Lemma~\ref{lemma:inequality} shows that the quantity on the right hand side of (\ref{equ:VI_gen1}) is always lower-bounded by zero. In the next two subsections, we will construct proper upper bounds for the quantity by replacing $(\boldsymbol{x},\boldsymbol{\lambda})$ with real estimate of PDMM. The algorithmic convergence will  be established by showing that the upper bounds approach to zero when iteration increases.

The conditions (\ref{equ:p_opt_cond2})-(\ref{equ:p_opt_cond4}) in Lemma~\ref{lemma:inequality} are not sufficient for showing that $(\boldsymbol{x},\boldsymbol{\lambda})$ is a saddle point of $L_{\mathcal{P}}$. The primal and dual feasibilities $\boldsymbol{x}\in X$ and $\boldsymbol{\lambda}\in \Lambda$ are also required to complete the argument, as shown in Lemma~\ref{lemma:primal_saddleP_converse}, \ref{lemma:dual_saddleP_converse} and \ref{lemma:mix_saddleP_converse} below. Lemma~\ref{lemma:primal_saddleP_converse} and \ref{lemma:dual_saddleP_converse} are preliminary to show that $(\boldsymbol{x},\boldsymbol{\lambda})$ is a saddle point of $L_{\mathcal{P}}$ as presented in Lemma~\ref{lemma:mix_saddleP_converse}. These three lemmas will be used in the next two subsections for convergence analysis.

\begin{lemma}
Let $(\boldsymbol{x}^{\star},\boldsymbol{\lambda}^{\star})$ be a saddle point of $L_{\mathcal{P}}$. 
Given $\boldsymbol{x}=\boldsymbol{x}'$ which satisfies (\ref{equ:p_opt_cond4}) and $\boldsymbol{x}'\in X$, then $(\boldsymbol{x}',\boldsymbol{\lambda}^{\star})$ is a saddle point of $L_{\mathcal{P}}$.
\label{lemma:primal_saddleP_converse}
\end{lemma}

\begin{proof}
By using (\ref{equ:p_opt_cond4}) and the fact that $\boldsymbol{x}'\in X$ and $\boldsymbol{\lambda}^{\star}\in \Lambda$, it is immediate from (\ref{equ:L_G_opt1})-(\ref{equ:L_G_opt4}) that $(\boldsymbol{x}',\boldsymbol{\lambda}^{\star})$ is a saddle point of $L_{\mathcal{P}}$.
\end{proof}

\begin{lemma}
Let $(\boldsymbol{x}^{\star},\boldsymbol{\lambda}^{\star})$ be a saddle point of $L_{\mathcal{P}}$.
Given $\boldsymbol{\lambda}=\boldsymbol{\lambda}'$ which satisfies (\ref{equ:p_opt_cond2}) and $\boldsymbol{\lambda}'\in \Lambda$, then $(\boldsymbol{x}^{\star},\boldsymbol{\lambda}')$ is a saddle point of $L_{\mathcal{P}}$.
\label{lemma:dual_saddleP_converse}
\end{lemma}
\begin{proof}
The proof is similar to that for Lemma~\ref{lemma:primal_saddleP_converse}. 
\end{proof}

\begin{lemma}
Let $(\boldsymbol{x}^{\star},\boldsymbol{\lambda}^{\star})$ be a saddle point of $L_{\mathcal{P}}$.
Given $(\boldsymbol{x},\boldsymbol{\lambda})=(\boldsymbol{x}',\boldsymbol{\lambda}')$ which satisfy (\ref{equ:p_opt_cond2})-(\ref{equ:p_opt_cond4}) and $(\boldsymbol{x}',\boldsymbol{\lambda}')\in (X,\Lambda)$, then $(\boldsymbol{x}',\boldsymbol{\lambda}')$ is a saddle point of $L_{\mathcal{P}}$.
\label{lemma:mix_saddleP_converse}
\end{lemma}
\begin{proof}
It is known from Lemma~\ref{lemma:primal_saddleP_converse} and \ref{lemma:dual_saddleP_converse} that in addition to $(\boldsymbol{x}^{\star},\boldsymbol{\lambda}^{\star})$,  $(\boldsymbol{x}',\boldsymbol{\lambda}^{\star})$ and  $(\boldsymbol{x}^{\star},\boldsymbol{\lambda}')$ are also the saddle points of $L_{\mathcal{P}}$. By using a similar argument as the one for Lemma~\ref{lemma:saddlePointSet}, one can show that $(\boldsymbol{x}',\boldsymbol{\lambda}')$ is a saddle point of $L_{\mathcal{P}}$.
\end{proof}


\subsection{Synchronous PDMM}
\label{subsec:rate_syn}

In this subsection, we show that the synchronous PDMM converges with the sub-linear rate $\mathcal{O}(K^{-1})$. In order to obtain the result, we need the following two lemmas.

\begin{lemma}
Let $(\boldsymbol{x}^{\star},\boldsymbol{\lambda}^{\star})$ be a saddle point of $L_{\mathcal{P}}$. The estimate $(\hat{\boldsymbol{x}}^{k+1},\hat{\boldsymbol{\lambda}}^{k+1})$ is obtained by performing (\ref{equ:x_updateSyn})-(\ref{equ:lambda_updateSyn}) under Condition~\ref{con:G}. Then there is
\begin{align}
&\hspace{-1mm}\sum_{i\in \mathcal{V}}\hspace{-0.5mm}\sum_{j\in \mathcal{N}_i}\hspace{-1.5mm}\Big[(\hat{\boldsymbol{\lambda}}_{i|j}^{k+1}\hspace{-0.5mm}-\hspace{-0.5mm}\boldsymbol{\lambda}_{i|j}^{\star})^T\Big(\hspace{-0.5mm}\boldsymbol{A}_{j i}\hat{\boldsymbol{x}}_j^{k+1}\hspace{-0.5mm}-\hspace{-0.5mm}\frac{\boldsymbol{c}_{ij}}{2}\Big)\hspace{-0.5mm}-\hspace{-0.5mm}(\hat{\boldsymbol{x}}_i^{k+1}\hspace{-0.5mm}-\hspace{-0.5mm}\boldsymbol{x}_i^{\star})^T\nonumber \\
&\hspace{8mm}\cdot\boldsymbol{A}_{i j}^T\hat{\boldsymbol{\lambda}}_{j|i}^{k+1}\Big]\hspace{-0.5mm}+\hspace{-0.5mm}p(\hat{\boldsymbol{x}}^{k+1},\hat{\boldsymbol{\lambda}}^{k+1})\leq \hspace{-0.6mm}\sum_{i\in \mathcal{V}}\hspace{-1mm}\sum_{j\in \mathcal{N}_i}d_{i|j}^{k+1},
\label{equ:optCon_syn}
\end{align}
where $d_{i|j}^{k+1}$ is given by
\begin{align}
&\hspace{-1mm}d_{i|j}^{k+1}\hspace{-0.5mm}=\hspace{-0.5mm}\frac{1}{2}\Big(\Big\|\boldsymbol{P}_{p,ij}^{\frac{1}{2}}\boldsymbol{A}_{j i}(\hat{\boldsymbol{x}}_j^{k}\hspace{-0.5mm}-\hspace{-0.5mm}\boldsymbol{x}_{j}^{\star})\hspace{-0.5mm}+\hspace{-0.5mm}\boldsymbol{P}_{p,ij}^{-\frac{1}{2}}(\boldsymbol{\lambda}_{j|i}^{\star}\hspace{-0.5mm}-\hspace{-0.5mm}\hat{\boldsymbol{\lambda}}_{j|i}^{k})\Big\|^2\nonumber \\
&\hspace{3mm}-\hspace{-0.5mm}\Big\|\boldsymbol{P}_{p,ij}^{\frac{1}{2}}\boldsymbol{A}_{j i}(\hat{\boldsymbol{x}}_j^{k+1}-\boldsymbol{x}_{j}^{\star})\hspace{-0.6mm}+\hspace{-0.6mm}\boldsymbol{P}_{p,ij}^{-\frac{1}{2}}(\boldsymbol{\lambda}_{j|i}^{\star}\hspace{-0.5mm}-\hspace{-0.5mm}\hat{\boldsymbol{\lambda}}_{j|i}^{k+1})\Big\|^2\nonumber \\ 
&\hspace{3mm}-\hspace{-0.5mm}\Big\|\boldsymbol{P}_{p,ij}^{\frac{1}{2}}(\boldsymbol{A}_{i j}\hat{\boldsymbol{x}}_i^{k+1}\hspace{-0.6mm}+\hspace{-0.6mm}\boldsymbol{A}_{j i}\hat{\boldsymbol{x}}_j^{k}\hspace{-0.6mm}-\hspace{-0.6mm}\boldsymbol{c}_{ij})\hspace{-0.6mm}+\hspace{-0.6mm}\boldsymbol{P}_{p,ij}^{-\frac{1}{2}}(\hat{\boldsymbol{\lambda}}_{i|j}^{k+1}\hspace{-0.6mm}-\hspace{-0.6mm}\hat{\boldsymbol{\lambda}}_{j|i}^{k})\Big\|^2\nonumber\\
&\hspace{3mm}+\hspace{-0.5mm}\|\Delta \boldsymbol{P}_{d,ij}^{\frac{1}{2}}(\boldsymbol{\lambda}_{j|i}^{\star}\hspace{-0.5mm}-\hspace{-0.5mm}\hat{\boldsymbol{\lambda}}_{j|i}^{k})\|^2\hspace{-1mm}-\hspace{-0.5mm}\|\Delta \boldsymbol{P}_{d,ij}^{\frac{1}{2}}(\boldsymbol{\lambda}_{j|i}^{\star}\hspace{-0.5mm}-\hspace{-0.5mm}\hat{\boldsymbol{\lambda}}_{j|i}^{k+1})\|^2\nonumber \\
&\hspace{3mm}-\hspace{-0.5mm}\|\Delta \boldsymbol{P}_{d,ij}^{\frac{1}{2}}(\hat{\boldsymbol{\lambda}}_{i|j}^{k+1}\hspace{-0.5mm}-\hspace{-0.5mm}\hat{\boldsymbol{\lambda}}_{j|i}^{k}) \|^2\Big),
\label{equ:d_syn}
\end{align}
where $\boldsymbol{P}_{p,ij}=\boldsymbol{P}_{p,ij}^{\frac{1}{2}}\boldsymbol{P}_{p,ij}^{\frac{1}{2}}$ and $\Delta \boldsymbol{P}_{d,ij}=\Delta\boldsymbol{P}_{d,ij}^{\frac{1}{2}}\Delta\boldsymbol{P}_{d,ij}^{\frac{1}{2}}$.
\label{lemma:gen_syn}
\end{lemma}

\begin{proof}See the proof in Appendix~\ref{appendix:lemma_gen_syn}. \end{proof}

\begin{lemma}
Every pair of estimates $(\hat{\boldsymbol{x}_i}^{k+1},\hat{\boldsymbol{\lambda}}_{i|j}^{k+1})$, $i\in\mathcal{V}$, $j\in \mathcal{N}_i$, $k\geq 0$, in Lemma~\ref{lemma:gen_syn} is upper bounded by a constant $M$ under a squared error criterion: 
\begin{align}
\Big\|\boldsymbol{P}_{p,ij}^{\frac{1}{2}}\boldsymbol{A}_{j i}(\hat{\boldsymbol{x}}_j^{k+1}-\boldsymbol{x}_{j}^{\star})\hspace{-0.6mm}+\hspace{-0.6mm}\boldsymbol{P}_{p,ij}^{-\frac{1}{2}}(\boldsymbol{\lambda}_{j|i}^{\star}\hspace{-0.5mm}-\hspace{-0.5mm}\hat{\boldsymbol{\lambda}}_{j|i}^{k+1})\Big\|^2\leq M.
\label{equ:bound_syn}
\end{align} 
\label{lemma:bound_syn}
\end{lemma}
\begin{proof}
One can first prove (\ref{equ:bound_syn}) for $k=0$ by performing algebra on (\ref{equ:optCon_syn})-(\ref{equ:d_syn}). The inequality (\ref{equ:bound_syn}) for $k>0$ can then be proved recursively. 
\end{proof}

Upon obtaining the results in Lemma~\ref{lemma:gen_syn} and \ref{lemma:bound_syn}, we are ready to present the convergence rate of synchronous PDMM.

\begin{sloppypar}
\begin{theorem}
Let $(\hat{\boldsymbol{x}}^{k},\hat{\boldsymbol{\lambda}}^{k})$, $k=1,\ldots,K$, be obtained by performing  (\ref{equ:x_updateSyn})-(\ref{equ:lambda_updateSyn}) under Condition~\ref{con:G}. 
The average estimate ${ (\bar{\boldsymbol{x}}^K,\bar{\boldsymbol{\lambda}}^K)=(\frac{1}{K}\sum_{k=1}^K\hat{\boldsymbol{x}}^k,\frac{1}{K}\sum_{k=1}^K\hat{\boldsymbol{\lambda}}^k)}$ satisfies 
\begin{align}
\hspace{-1.5mm}0\leq&\hspace{-0.6mm}\sum_{i\in \mathcal{V}}\hspace{-0.8mm}\sum_{j\in \mathcal{N}_i}\hspace{-1.8mm}\Big[(\bar{\boldsymbol{\lambda}}_{i|j}^{K}\hspace{-0.5mm}-\hspace{-0.5mm}\boldsymbol{\lambda}_{i|j}^{\star})^T\hspace{-0.6mm}\Big(\hspace{-0.5mm}\boldsymbol{A}_{j i}\bar{\boldsymbol{x}}_j^{K}\hspace{-0.8mm}-\hspace{-0.8mm}\frac{\boldsymbol{c}_{ij}}{2}\Big)\hspace{-0.8mm}-\hspace{-0.6mm}(\bar{\boldsymbol{x}}_i^{K}\hspace{-0.5mm}-\hspace{-0.5mm}\boldsymbol{x}_i^{\star})^T\nonumber\\
\hspace{-1.5mm}&\hspace{13mm}\cdot\boldsymbol{A}_{i j}^T\bar{\boldsymbol{\lambda}}_{j|i}^{K}\Big]\hspace{-0.5mm}+\hspace{-0.5mm}p(\bar{\boldsymbol{x}}^{K},\bar{\boldsymbol{\lambda}}^{K})\leq \mathcal{O}\Big(\frac{1}{K}\Big)\label{equ:optCon_syn_sublinear2}
\end{align}
\begin{align}
\lim_{K\rightarrow \infty}&\Big[\boldsymbol{P}_{p,ij}^{\frac{1}{2}}(\boldsymbol{A}_{i j}\bar{\boldsymbol{x}}_i^{K}+\boldsymbol{A}_{j i}\bar{\boldsymbol{x}}_j^{K}-\boldsymbol{c}_{ij})\nonumber\\
&\hspace{1mm}+\boldsymbol{P}_{p,ij}^{-\frac{1}{2}}(\bar{\boldsymbol{\lambda}}_{i|j}^{K}-\bar{\boldsymbol{\lambda}}_{j|i}^{K})\Big]=\boldsymbol{0} \hspace{12.5mm} \forall [i,j]\in \vec{\mathcal{E}}.
\label{equ:mix_feas_syn}
\end{align}
\label{theo:syn}
\vspace{-2mm}
\end{theorem}
\end{sloppypar}
\begin{proof}
Summing (\ref{equ:optCon_syn}) over $k$ and simplifying the expression yields
\begin{align}
&\hspace{-1mm}\sum_{k=0}^{K-1}\hspace{-0.5mm}\Bigg(\sum_{i\in \mathcal{V}}\hspace{-0.5mm}\sum_{j\in \mathcal{N}_i}\hspace{-1.5mm}\Big[(\hat{\boldsymbol{\lambda}}_{i|j}^{k+1}\hspace{-0.5mm}-\hspace{-0.5mm}\boldsymbol{\lambda}_{i|j}^{\star})^T\Big(\hspace{-0.5mm}\boldsymbol{A}_{j i}\hat{\boldsymbol{x}}_j^{k+1}\hspace{-0.5mm}-\hspace{-0.5mm}\frac{\boldsymbol{c}_{ij}}{2}\Big)\hspace{-0.5mm}\nonumber\\
&\hspace{8mm}-\hspace{-0.5mm}(\hat{\boldsymbol{x}}_i^{k+1}\hspace{-0.5mm}-\hspace{-0.5mm}\boldsymbol{x}_i^{\star})^T\boldsymbol{A}_{i j}^T\hat{\boldsymbol{\lambda}}_{j|i}^{k+1}\Big]\hspace{-0.5mm}+\hspace{-0.5mm}p(\hat{\boldsymbol{x}}^{k+1},\hat{\boldsymbol{\lambda}}^{k+1})\hspace{-0.5mm}+\hspace{-0.5mm}\sum_{i\in\mathcal{V}}\hspace{-0.5mm}\sum_{j\in \mathcal{N}}\nonumber\\
&\Big[\Big\|\boldsymbol{P}_{p,ij}^{\frac{1}{2}}(\boldsymbol{A}_{i j}\hat{\boldsymbol{x}}_i^{k+1}\hspace{-0.5mm}+\hspace{-0.5mm}\boldsymbol{A}_{j i}\hat{\boldsymbol{x}}_j^{k}\hspace{-0.5mm}-\hspace{-0.5mm}\boldsymbol{c}_{ij})\hspace{-0.5mm}+\hspace{-0.5mm}\boldsymbol{P}_{p,ij}^{-\frac{1}{2}}(\hat{\boldsymbol{\lambda}}_{i|j}^{k+1}\hspace{-0.5mm}-\hspace{-0.5mm}\hat{\boldsymbol{\lambda}}_{j|i}^{k})\Big\|^2\Big]\Bigg)\nonumber\\
&\hspace{-1.5mm}\leq \hspace{-0.5mm}\sum_{i\in \mathcal{V}}\hspace{-1mm}\sum_{j\in \mathcal{N}_i}\hspace{-0.5mm}\frac{1}{2}\Big(\Big\|\boldsymbol{P}_{p,ij}^{\frac{1}{2}}(\boldsymbol{A}_{j i}(\hat{\boldsymbol{x}}_j^{0}\hspace{-0.5mm}-\hspace{-0.5mm}{\boldsymbol{x}}_j^{\star} )\hspace{-0.5mm}+\hspace{-0.5mm}\boldsymbol{P}_{p,ij}^{-\frac{1}{2}}({\boldsymbol{\lambda}}_{i|j}^{\star}\hspace{-0.5mm}-\hspace{-0.5mm}\hat{\boldsymbol{\lambda}}_{j|i}^{0})\Big\|^2\nonumber\\
&\hspace{21mm}+\hspace{-0.5mm}\|\Delta \boldsymbol{P}_{d,ij}^{\frac{1}{2}}({\boldsymbol{\lambda}}_{j|i}^{\star}\hspace{-0.5mm}-\hspace{-0.5mm}\hat{\boldsymbol{\lambda}}_{j|i}^{0}) \|^2\Big).
\label{equ:optCon_syn_sublinear}
\end{align}
Finally, since the left hand side of (\ref{equ:optCon_syn_sublinear}) is a convex function of $(\boldsymbol{x},\boldsymbol{\lambda})$, applying Jensen's inequality to (\ref{equ:optCon_syn_sublinear}) and using the inequality of Lemma~\ref{lemma:inequality} yields (\ref{equ:optCon_syn_sublinear2}). Similarly, applying Jensen's inequality to (\ref{equ:optCon_syn_sublinear}) and using the upper-bound result of Lemma~\ref{lemma:bound_syn} yields the asymptotic result (\ref{equ:mix_feas_syn}).
\end{proof}

Finally, we use the results of Theorem~\ref{theo:syn} to show that as $K$ goes to infinity, the average estimate $(\bar{\boldsymbol{x}}^{K},\bar{\boldsymbol{\lambda}}^{K})$ converges to a saddle point  of $L_{\mathcal{P}}$.

\begin{theorem}
The average estimate $(\bar{\boldsymbol{x}}^{K},\bar{\boldsymbol{\lambda}}^{K})$ of Theorem~\ref{theo:syn} converges to a saddle point $(\boldsymbol{x}^{\star},\boldsymbol{\lambda}^{\star})$ of $L_{\mathcal{P}}$ as $K$ increases.
\label{theo:syn1}\vspace{-2mm}
\end{theorem}
\begin{proof}
The basic idea of the proof is to investigate if $(\bar{\boldsymbol{x}}^{K},\bar{\boldsymbol{\lambda}}^{K})$ satisfies all the conditions of Lemma~\ref{lemma:mix_saddleP_converse}. 
By investigation of Lemma~\ref{lemma:inequality} and (\ref{equ:optCon_syn_sublinear2}), it is clear that the average estimate $(\bar{\boldsymbol{x}}^{K},\bar{\boldsymbol{\lambda}}^{K})$ asymptotically satisfies the conditions (\ref{equ:p_opt_cond2})-(\ref{equ:p_opt_cond4}) by letting $(\boldsymbol{x},\boldsymbol{\lambda})=(\bar{\boldsymbol{x}}^{K},\bar{\boldsymbol{\lambda}}^{K})$.

Next we show that as $K$ increases, $\bar{\boldsymbol{x}}^{K}$ asymptotically converges to an element of the primal feasible set $X$ and so does $\bar{\boldsymbol{\lambda}}^{K}$ to an element of the dual feasible set $\Lambda$. To do do, we reconsider (\ref{equ:mix_feas_syn}) for each pair of directed edges $[i,j]$ and $[j,i]$, which can be expressed as
\begin{align}
\lim_{K\rightarrow \infty}&\Big[\boldsymbol{P}_{p,ij}^{\frac{1}{2}}(\boldsymbol{A}_{i j}\bar{\boldsymbol{x}}_i^{K}\hspace{-0.6mm}+\hspace{-0.6mm}\boldsymbol{A}_{j i}\bar{\boldsymbol{x}}_j^{K}\hspace{-0.6mm}-\hspace{-0.6mm}\boldsymbol{c}_{ij})\hspace{-0.6mm}+\hspace{-0.6mm}\boldsymbol{P}_{p,ij}^{-\frac{1}{2}}(\bar{\boldsymbol{\lambda}}_{i|j}^{K}\hspace{-0.6mm}-\hspace{-0.6mm}\bar{\boldsymbol{\lambda}}_{j|i}^{K})\Big]\hspace{-0.6mm}=\hspace{-0.6mm}\boldsymbol{0} \nonumber \\
\lim_{K\rightarrow \infty}&\Big[\boldsymbol{P}_{p,ij}^{\frac{1}{2}}(\boldsymbol{A}_{i j}\bar{\boldsymbol{x}}_i^{K}\hspace{-0.6mm}+\hspace{-0.6mm}\boldsymbol{A}_{j i}\bar{\boldsymbol{x}}_j^{K}\hspace{-0.6mm}-\hspace{-0.6mm}\boldsymbol{c}_{ij})\hspace{-0.6mm}+\hspace{-0.6mm}\boldsymbol{P}_{p,ij}^{-\frac{1}{2}}(\bar{\boldsymbol{\lambda}}_{j|i}^{K}\hspace{-0.6mm}-\hspace{-0.6mm}\bar{\boldsymbol{\lambda}}_{i|j}^{K})\Big]\hspace{-0.6mm}=\hspace{-0.6mm}\boldsymbol{0}. \nonumber
\end{align}
Combining the above two expressions produces
\begin{align}
\lim_{K\rightarrow \infty}&\boldsymbol{A}_{ij}\bar{\boldsymbol{x}}_i^{K}+\boldsymbol{A}_{ji}\bar{\boldsymbol{x}}_j^{K}=\boldsymbol{c}_{ij} \hspace{11mm}\forall (i,j)\in \mathcal{E} \nonumber\\
\lim_{K\rightarrow \infty}&\bar{\boldsymbol{\lambda}}_{j|i}^{K}=\bar{\boldsymbol{\lambda}}_{i|j}^{K} \hspace{30mm} \forall (i,j)\in \mathcal{E}.\nonumber
\end{align}
It is straightforward from Lemma~\ref{lemma:mix_saddleP_converse} that $(\bar{\boldsymbol{x}}^{K},\bar{\boldsymbol{\lambda}}^{K})$ converges to a saddle point of $L_{\mathcal{P}}$ as $K$ increases.\vspace{0.5mm}
\end{proof}

Further we have the following result from Theorem~\ref{theo:syn1}:
\begin{corollary}
If for certain $i\in \mathcal{V}$, the estimate $\hat{\boldsymbol{x}}_i^{k}$ in Theorem~\ref{theo:syn} converges to a fixed point $\boldsymbol{x}_i'$ ($\lim_{k\rightarrow \infty}\hat{\boldsymbol{x}}_i^{k}=\boldsymbol{x}_i'$), we have $\boldsymbol{x}_i'=\boldsymbol{x}_i^{\star}$ which is the $i$th component of the optimal solution $\boldsymbol{x}^{\star}$ in Theorem~\ref{theo:syn1}. Similarly, if the estimate $\hat{\boldsymbol{\lambda}}_{i|j}^{k}$ converges to a point $\boldsymbol{\lambda}_{i|j}'$, we have $\boldsymbol{\lambda}_{i|j}'=\boldsymbol{\lambda}_{i|j}^{\star}$.  
\label{coro:syn}
\end{corollary}

\subsection{Asynchronous PDMM}
\label{subsec:rate_asyn}
In this subsection, we characterize the convergence rate of asynchronous PDMM. In order to facilitate the analysis, we consider a predefined  node-activation strategy (no randomness is involved). We suppose at each iteration $k$, the node $i=\textrm{mod}(k,m)+1$ is activated for parameter-updating, where $m=|\mathcal{V}|$ and $\textrm{mod}(\cdot,\cdot)$ stands for the modulus operation. Then naturally, after a segment of $m$ consecutive iterations, all the nodes will be activated sequentially, one node at each iteration.

To be able to derive the convergence rate, we consider segments of iterations, i.e., $k\in \{lm,lm+1,\ldots (l+1)m-1\}$, where $l\geq 0$. Each segment $l$ consists of $m$ iterations. With the mapping $i=\textrm{mod}(k,m)+1$, it is immediate that $k=ml$ activates node~1 and $k=(l+1)m-1$ activates node~$m$. Based on the above analysis, we have the following result.
\begin{lemma}
Let $k_1,k_2$ be two iteration indices within a segment $\{lm,lm+1,\ldots, (l+1)m-1\}$. If $k_1<k_2$, then $i_1<i_2$, where the node-index $i_q=\textrm{mod}(k_q,m)+1$, $q=1,2$.
\label{lemma:act_node_order}
\end{lemma}

Upon introducing Lemma~\ref{lemma:act_node_order}, we are ready to perform convergence analysis.
\begin{lemma}
Let $(\boldsymbol{x}^{\star},\boldsymbol{\lambda}^{\star})$ be a saddle point of $L_{\mathcal{P}}$. A segment of estimates $\{(\hat{\boldsymbol{x}}^{k+1},\hat{\boldsymbol{\lambda}}^{k+1})|k=lm,\ldots,(l+1)m-1 \}$, is obtained by performing (\ref{equ:x_lambda_updateAsyn1})-(\ref{equ:x_lambda_updateAsyn2}) under Condition~\ref{con:G}. Then there is
\begin{align}
&\hspace{-3mm}\sum_{i\in \mathcal{V}}\hspace{-0.8mm}\sum_{j\in \mathcal{N}_i}\hspace{-1.2mm}\Big[\hspace{-1.5mm}\left(\hspace{-0.6mm}\hat{\boldsymbol{\lambda}}_{i|j}^{(l\hspace{-0.3mm}+\hspace{-0.4mm}1)m}\hspace{-0.9mm}-\hspace{-0.9mm}\boldsymbol{\lambda}_{i|j}^{\star}\hspace{-0.6mm}\right)^{\hspace{-0.6mm}T}\hspace{-2mm}\Big(\hspace{-0.5mm}\boldsymbol{A}_{j i}\hat{\boldsymbol{x}}_j^{(l\hspace{-0.4mm}+\hspace{-0.4mm}1)m}\hspace{-0.7mm}-\hspace{-0.7mm}\frac{\boldsymbol{c}_{ij}}{2}\hspace{-0.5mm}\Big)\hspace{-1mm}-\hspace{-1mm}\Big(\hspace{-0.5mm}\hat{\boldsymbol{x}}_i^{(l\hspace{-0.3mm}+\hspace{-0.3mm}1)m}\hspace{-0.7mm}-\hspace{-0.7mm}\boldsymbol{x}_i^{\star}\hspace{-0.5mm}\Big)^{\hspace{-0.6mm}T}\hspace{-0.5mm}\nonumber\\
&\hspace{0mm}\cdot\boldsymbol{A}_{i j}^T\hat{\boldsymbol{\lambda}}_{j|i}^{(l+1)m}\Big]\hspace{-0.5mm}+\hspace{-0.5mm}p\left(\hat{\boldsymbol{x}}^{(l+1)m},\hat{\boldsymbol{\lambda}}^{(l+1)m}\right)\leq \hspace{-1mm}\sum_{(u,v)\in \mathcal{E}}^{u<v}\hspace{-1mm}d_{uv}^{l+1}, \label{equ:optCon_gen_asyn}
\end{align}
where $d_{uv}^{l+1}$ is given by
\begin{align}
&\hspace{0mm}d_{uv}^{l+1}\hspace{-0.5mm}=\hspace{-0.5mm}\frac{1}{2}\Big(\|\boldsymbol{P}_{p,uv}^{\frac{1}{2}}\boldsymbol{A}_{v u}(\hat{\boldsymbol{x}}_v^{lm}-\boldsymbol{x}_{v}^{\star} )+\boldsymbol{P}_{p,uv}^{-\frac{1}{2}}(\boldsymbol{\lambda}_{v|u}^{\star}-\hat{\boldsymbol{\lambda}}_{v|u}^{lm})\|^2\nonumber\\
&\hspace{3mm}-\hspace{-0.5mm}\|\boldsymbol{P}_{p,uv}^{\frac{1}{2}}\boldsymbol{A}_{v u}\hspace{-0.15mm}(\hat{\boldsymbol{x}}_v^{(l+1)m}\hspace{-0.5mm}-\hspace{-0.5mm}\boldsymbol{x}_{v}^{\star})\hspace{-0.5mm}+\hspace{-0.5mm}\boldsymbol{P}_{p,uv}^{-\frac{1}{2}}\hspace{-0.15mm}(\boldsymbol{\lambda}_{v|u}^{\star}\hspace{-0.6mm}-\hspace{-0.6mm}\hat{\boldsymbol{\lambda}}_{v|u}^{(l+1)m})\|^2\nonumber \\
&\hspace{3mm}-\|\boldsymbol{P}_{p,uv}^{\frac{1}{2}}(\boldsymbol{A}_{u v}\hat{\boldsymbol{x}}_u^{(l+1)m}+\boldsymbol{A}_{vu}\hat{\boldsymbol{x}}_v^{(l+1)m}-\boldsymbol{c}_{uv})\nonumber\\
&\hspace{8mm}-\boldsymbol{P}_{p,uv}^{-\frac{1}{2}}(\hat{\boldsymbol{\lambda}}_{u|v}^{(l+1)m}-\hat{\boldsymbol{\lambda}}_{v|u}^{(l+1)m})\|^2\nonumber\\
&\hspace{3mm}-\|\boldsymbol{P}_{p,uv}^{\frac{1}{2}}(\boldsymbol{A}_{u v}\hat{\boldsymbol{x}}_u^{(l+1)m}+\boldsymbol{A}_{v u}\hat{\boldsymbol{x}}_v^{lm}-\boldsymbol{c}_{uv})\nonumber\\
&\hspace{8mm}+\hspace{-0.5mm}\boldsymbol{P}_{p,uv}^{-\frac{1}{2}}(\hat{\boldsymbol{\lambda}}_{u|v}^{(l\hspace{-0.3mm}+\hspace{-0.3mm}1)m}\hspace{-0.6mm}-\hspace{-0.6mm}\hat{\boldsymbol{\lambda}}_{v|u}^{lm})\|^2\hspace{-0.6mm}+\hspace{-0.6mm}\|\Delta\boldsymbol{P}_{d,uv}^{\frac{1}{2}}({\boldsymbol{\lambda}}_{u|v}^{\star}\hspace{-0.6mm}-\hspace{-0.6mm}\hat{\boldsymbol{\lambda}}_{v|u}^{lm})\|^2\nonumber\\
&\hspace{3mm}-\hspace{-0.7mm}\|\Delta\boldsymbol{P}_{d,uv}^{\frac{1}{2}}\hspace{-0.1mm}(\hspace{-0.1mm}{\boldsymbol{\lambda}}_{u|v}^{\star}\hspace{-0.7mm}-\hspace{-0.7mm}\hat{\boldsymbol{\lambda}}_{v|u}^{(l\hspace{-0.3mm}+\hspace{-0.3mm}1)m}\hspace{-0.1mm})\hspace{-0.1mm}\|^2\hspace{-0.7mm}-\hspace{-0.7mm}\|\hspace{-0.1mm}\Delta\boldsymbol{P}_{d,uv}^{\frac{1}{2}}(\hspace{-0.1mm}\hat{\boldsymbol{\lambda}}_{u|v}^{(l\hspace{-0.3mm}+\hspace{-0.3mm}1)m}\hspace{-0.8mm}-\hspace{-0.8mm}\hat{\boldsymbol{\lambda}}_{v|u}^{lm}\hspace{-0.1mm})\hspace{-0.1mm}\|^2\nonumber\\
&\hspace{3mm}-\hspace{-0.5mm}\|\Delta\boldsymbol{P}_{d,uv}^{\frac{1}{2}}(\hat{\boldsymbol{\lambda}}_{u|v}^{(l+1)m}\hspace{-0.5mm}-\hspace{-0.5mm}\hat{\boldsymbol{\lambda}}_{v|u}^{(l+1)m})\|^2\Big)\qquad u<v. \label{equ:d_asyn}
\end{align}
\label{lemma:gen_asyn}\vspace{-4mm}
\end{lemma}
\begin{proof}
See the proof in Appendix~\ref{appendix:lemma_gen_asyn}. Lemma~\ref{lemma:act_node_order} will be used in the proof to simplify mathematic derivations.
\end{proof}
\begin{remark}
We note that Lemma~\ref{lemma:gen_asyn} corresponds to Lemma~\ref{lemma:gen_syn} which is for synchronous PDMM.
The right hand side of (\ref{equ:optCon_gen_asyn}) consists of $|\mathcal{E}|$ quantities $\{d_{uv}^{l+1}\}$ (one for each edge $(u,v)\in\mathcal{E}$) as opposed to that of (\ref{equ:optCon_syn}) which consists of $|\vec{\mathcal{E}}|$ quantities $\{d_{i|j}^{k+1}\}$ (one for each directed edge $[i,j]\in\vec{\mathcal{E}}$).
\end{remark}
\begin{lemma}
Every pair of estimates $(\hat{\boldsymbol{x}}_v^{(l+1)m},\hat{\boldsymbol{\lambda}}_{v|u}^{(l+1)m})$, $(u,v)\in\mathcal{E}$, $u<v$, $l\geq 0$, in Lemma~\ref{lemma:gen_asyn} is upper bounded by a constant $M$ under a squared error criterion: 
\begin{align}
&\hspace{-4mm}\|\boldsymbol{P}_{\hspace{-0.5mm}p,uv}^{\frac{1}{2}}\boldsymbol{A}_{v u}\hspace{-0.15mm}(\hat{\boldsymbol{x}}_v^{(l\hspace{-0.3mm}+\hspace{-0.3mm}1)m}\hspace{-0.6mm}-\hspace{-0.6mm}\boldsymbol{x}_{v}^{\star})\hspace{-0.7mm}+\hspace{-0.7mm}\boldsymbol{P}_{\hspace{-0.3mm}p,uv}^{-\frac{1}{2}}\hspace{-0.15mm}(\boldsymbol{\lambda}_{v|u}^{\star}\hspace{-0.8mm}-\hspace{-0.8mm}\hat{\boldsymbol{\lambda}}_{v|u}^{(l\hspace{-0.3mm}+\hspace{-0.3mm}1)m})\|^2\hspace{-0.9mm}\leq\hspace{-0.6mm} M. \nonumber
\end{align} 
\label{lemma:bound_asyn}
 \vspace{-4mm}
\end{lemma}

\begin{theorem}
Let the $K\geq 1$ segments of estimates $\{(\hat{\boldsymbol{x}}^{k+1},\hat{\boldsymbol{\lambda}}^{k+1})| k=0,\ldots, Km-1\}$ be obtained by performing  (\ref{equ:x_lambda_updateAsyn1})-(\ref{equ:x_lambda_updateAsyn2}) under Condition~\ref{con:G}. The average estimates $(\check{\boldsymbol{x}}^K,\check{\boldsymbol{\lambda}}^K)\hspace{-0.5mm}=\hspace{-0.5mm}(\frac{1}{K}\hspace{-0.3mm}\sum_{l=1}^K\hspace{-0.3mm}\hat{\boldsymbol{x}}^{lm},\frac{1}{K}\hspace{-0.3mm}\sum_{l=1}^K\hspace{-0.3mm}\hat{\boldsymbol{\lambda}}^{lm})$ satisfies 
\begin{align}
&\hspace{-2mm}0\hspace{-0.5mm}\leq\hspace{-0.5mm}\sum_{i\in \mathcal{V}}\hspace{-1mm}\sum_{j\in \mathcal{N}_i}\hspace{-1.5mm}\Big[\hspace{-0.5mm}\left(\check{\boldsymbol{\lambda}}_{i|j}^{K}\hspace{-0.5mm}-\hspace{-0.5mm}\boldsymbol{\lambda}_{i|j}^{\star}\right)^T\hspace{-1.2mm}\Big(\hspace{-0.5mm}\boldsymbol{A}_{j i}\check{\boldsymbol{x}}_j^{K}\hspace{-0.5mm}-\hspace{-0.5mm}\frac{\boldsymbol{c}_{ij}}{2}\Big)\hspace{-0.6mm}-\hspace{-0.6mm}\Big(\hspace{-0.2mm}\check{\boldsymbol{x}}_i^{K}\hspace{-0.5mm}-\hspace{-0.5mm}\boldsymbol{x}_i^{\star}\hspace{-0.2mm}\Big)^T\nonumber\\
&\hspace{14mm}\cdot\boldsymbol{A}_{i j}^T\check{\boldsymbol{\lambda}}_{j|i}^{K}\Big]\hspace{-0.5mm}+\hspace{-0.5mm}p\left(\check{\boldsymbol{x}}^{K},\check{\boldsymbol{\lambda}}^{K}\right)\leq \mathcal{O}\Big(\frac{1}{K}\Big) \\
&\hspace{-2mm}0\hspace{-0.5mm}\leq\Big\|\boldsymbol{P}_{p,uv}^{\frac{1}{2}}(\boldsymbol{A}_{u v}\check{\boldsymbol{x}}_u^{K}+\boldsymbol{A}_{v u}\check{\boldsymbol{x}}_v^{K}-\boldsymbol{c}_{uv})\nonumber\\
&\hspace{2mm}-\hspace{-0.5mm}\boldsymbol{P}_{p,uv}^{-\frac{1}{2}}(\check{\boldsymbol{\lambda}}_{u|v}^{K}\hspace{-0.6mm}-\hspace{-0.6mm}\check{\boldsymbol{\lambda}}_{v|u}^{K})\Big\|^2\hspace{-0.6mm}\leq\hspace{-0.5mm}\mathcal{O}\Big(\frac{1}{K}\Big)
\hspace{1.5mm} \forall (u,v)\hspace{-0.5mm}\in\hspace{-0.5mm} \mathcal{E}, u\hspace{-0.5mm}<\hspace{-0.5mm}v \\
&\hspace{-2mm}\lim_{K \rightarrow \infty}\hspace{-0.2mm}\Big[\boldsymbol{P}_{p,uv}^{\frac{1}{2}}(\boldsymbol{A}_{u v}\check{\boldsymbol{x}}_u^{K}+\boldsymbol{A}_{v u}\check{\boldsymbol{x}}_v^{K}-\boldsymbol{c}_{uv})\nonumber\\
&\hspace{2mm}+\hspace{-0.5mm}\boldsymbol{P}_{p,uv}^{-\frac{1}{2}}(\check{\boldsymbol{\lambda}}_{u|v}^{K}\hspace{-0.5mm}-\hspace{-0.5mm}\check{\boldsymbol{\lambda}}_{v|u}^{K})\Big]\hspace{-0.5mm}=\hspace{-0.5mm}\boldsymbol{0}\hspace{1mm} \forall (u,v)\hspace{-0.5mm}\in\hspace{-0.5mm} \mathcal{E}, u\hspace{-0.5mm}<\hspace{-0.5mm}v.\label{equ:mix_feas_asyn2}
\end{align}
\label{theo:sublinear_asyn}
\vspace{-4mm}
\end{theorem}
\begin{proof}
The proof is similar to that for Theorem~\ref{theo:syn}. \vspace{0.5mm}
\end{proof}

Similarly to synchrounous PDMM, by using the results of Theorem~\ref{theo:sublinear_asyn}, we can conclude that: 
\begin{theorem}
The average estimate $(\check{\boldsymbol{x}}^{K},\check{\boldsymbol{\lambda}}^{K})$ of Theorem~\ref{theo:sublinear_asyn} converges to a saddle point $(\boldsymbol{x}^{\star},\boldsymbol{\lambda}^{\star})$ of $L_{\mathcal{P}}$ as $K$ increases.
\label{theo:asyn1}
\end{theorem}

\begin{corollary}
If for certain $u \in \mathcal{V}$, the estimate $\hat{\boldsymbol{x}}_u^{lm}$ in Theorem~\ref{theo:sublinear_asyn} converges to a fixed point $\boldsymbol{x}_u'$ ($\lim_{l\rightarrow \infty} \hat{\boldsymbol{x}}_u^{lm}= \boldsymbol{x}_u'$), we have ${\boldsymbol{x}}_u'=\boldsymbol{x}_u^{\star}$ which is the $u$th component of the optimal solution $\boldsymbol{x}^{\star}$ in Theorem~\ref{theo:asyn1}. Similarly, if the estimate $\hat{\boldsymbol{\lambda}}_{u|v}^{lm}$ converges to a point $\boldsymbol{\lambda}_{u|v}'$, we hvae $\boldsymbol{\lambda}_{u|v}'=\boldsymbol{\lambda}_{u|v}^{\star}$.
\label{coro:asyn}
\end{corollary}
   
\section{Application to Distributed Averaging}
\label{sec:disAve}

In this section, we consider solving the problem of distributed averaging by using PDMM. Distributed averaging is one of the basic and important operations for advanced distributed signal processing \cite{Boyd06gossip,Dimakis10GossipAlg}. 


\subsection{Problem formulation}
Suppose every node $i$ in a graph $G=(\mathcal{V},\mathcal{E})$ carries a scalar parameter, denoted as $t_i$. $t_i$ may represent a measurement of the environment, such as temperature, humidity or darkness. The problem is to compute the average value $t_{ave}=\frac{1}{m}\sum_{i\in \mathcal{V}}t_i$ iteratively only through message-passing between neighboring nodes in the graph. 

The above averaging problem can be formulated as a quadratic optimization over the graph as
\begin{align}
\min_{\{x_i\}}\sum_{i\in \mathcal{V}}\frac{1}{2}(x_i-t_i)^2\quad \textrm{ s.t. } x_i-x_j=0 \quad \forall (i,j)\in \mathcal{E}.\label{equ:avePrim}
\end{align}
The optimal solution equals to $x_1^{\star}=\ldots=x_m^{\star}=t_{ave}$, which is the same as the averaging value. 

The quadratic problem (\ref{equ:avePrim}) is inline with (\ref{equ:optProMulti_re}) by letting
\begin{align}
&f_i(x_i)=\frac{1}{2}(x_i-t_i)^2\quad \forall i\in \mathcal{V} \label{equ:avePrim1}\\
&(\boldsymbol{A}_{i j},\boldsymbol{A}_{j i},\boldsymbol{c}_{ij})=(1,-1,0)\quad \forall(i,j)\in \mathcal{E},i<j. \label{equ:avePrim2}
\end{align} In next subsection, we apply PDMM for distributed averaging.

\subsection{Parameter computations and transmissions}
Before deriving the updating expressions for PDMM, we first configure the set $\mathcal{P}$ in $L_{\mathcal{P}}$. For distributed averaging, all the matrices in $\mathcal{P}$ become scalars. For simplicity, we set the value of the primal scalars and the dual scalars as
\begin{subequations}
\begin{align}\boldsymbol{P}_{p,ij}&=\gamma_p\quad \forall(i,j)\in \mathcal{E} \\
\boldsymbol{P}_{d,ij}&=\gamma_d \quad \forall (i,j)\in \mathcal{E},
\end{align}
 \label{equ:ave_parSet}
\end{subequations} where the two parameters $\gamma_p>0$ and $\gamma_d>0$.

We start with the synchronous PDMM. 
By inserting (\ref{equ:avePrim1})-(\ref{equ:ave_parSet}) into (\ref{equ:x_updateSyn}), (\ref{equ:lambda_x_relation1}) and (\ref{equ:w_lambda_optCond5}), the updating expression for $(\hat{\boldsymbol{x}}^{k+1},\hat{\boldsymbol{\lambda}}^{k+1})$ at iteration $k$ can be derived as
\begin{align}
\hat{x}_i^{k+1}\hspace{-0.4mm}&=\hspace{-0.4mm}\frac{ t_i+\sum_{j\in \mathcal{N}_i}(\gamma_p\hat{x}_j^k+\boldsymbol{A}_{i j}\hat{\lambda}_{j|i}^k)}{1+|\mathcal{N}_i|\gamma_p}\quad \forall i\in \mathcal{V} \label{equ:ave_x}\\
\hat{\lambda}_{i|j}^{k+1}\hspace{-0.4mm}&=\hspace{-0.4mm}\hat{\lambda}_{j|i}^k\hspace{-0.5mm}-\hspace{-0.5mm}\frac{1}{\gamma_d}\Big(\boldsymbol{A}_{j i}\hat{x}_j^{k}\hspace{-0.5mm}+\hspace{-0.5mm}\boldsymbol{A}_{i j}w_i^{k+1}\Big)\quad \forall [i,j]\in \vec{\mathcal{E}},\label{equ:ave_lambda}
\end{align}
where 
\begin{align}
\hspace{-2mm}w_i^{k+1}\hspace{-0.4mm}=\hspace{-0.4mm}\frac{\sum_{j\in \mathcal{N}_i}(\hat{x}_j^k\hspace{-0.3mm}+\hspace{-0.3mm}\gamma_d\boldsymbol{A}_{i j}\hat{\lambda}_{j|i}^k)\hspace{-0.3mm}+\hspace{-0.3mm}\gamma_dt_i}{|\mathcal{N}_i|+\gamma_d}\quad \forall i\in \mathcal{V}. \label{equ:ave_w}
\end{align}
For the case that $\gamma_d=\gamma_p^{-1}$, it is immediate from (\ref{equ:ave_x}) and (\ref{equ:ave_w}) that $\hat{x}_i^{k+1}={w}_i^{k+1}$, which coincides with Proposition~\ref{prop:lambda_x_relation}.

The asynchronous PDMM only activates one node per iteration. Suppose node $i$ is activated at iteration $k$. Node $i$ then updates $\hat{{x}}_i$ and $\hat{{\lambda}}_{i|j}$, $j\in \mathcal{N}_i$, by following (\ref{equ:ave_x})-(\ref{equ:ave_lambda}) while all other nodes remain silent. After computation, node $i$ then sends $(\hat{x}_i,\hat{\lambda}_{i|j})$ to its neighboring node $j$ for all neighbors. 

As described in Subsection~\ref{subsec:node_comp_trans}, if no transmission fails in the graph, the transmission of $\hat{\lambda}_{i|j}$, $j\in \mathcal{N}_i$, can be replaced by broadcast transmission of $w_i$ as given by (\ref{equ:ave_w}).
Once $w_i$ is received by a neighboring node $j$, $\hat{\lambda}_{i|j}$ can be easily computed by node $j$ alone using $w_i$, $\hat{x}_j$ and $\hat{\lambda}_{j|i}$ (see Eq.~(\ref{equ:ave_lambda})). If instead the transmission is not reliable, we have to return to point-to-point transmission.  


\subsection{Experimental results}
\label{subsec:experiment}
We conducted three experiments for PDMM applied to distributed averaging. In the first experiment, we evaluated how different parameter-settings w.r.t. $(\gamma_p,\gamma_d)$ affect the convergence rates of PDMM. In the second experiment, we tested the non-perfect channels for PDMM, which lacks convergence guaranty at the moment. Finally, we evaluated the convergence rates of PDMM, ADMM and two gossip algorithms.


The tested graph in the three experiments was a $10\times 10$ two-dimensional grid (corresponding to $m=100$), implying that each node may have two, three or four neighbors. The mean squared error (MSE) $\frac{1}{m}\|\hat{\boldsymbol{x}}-t_{ave}\boldsymbol{1}\|_2^2$ was employed as performance measurement.

\subsubsection{performance for different parameter settings}
\label{subsub:parameterSel}
In this experiment, we evaluated the performance of PDMM by testing different parameter-settings for $(\gamma_p,\gamma_d)$. Both synchronous and asynchronous updating schemes were investigated.

At each iteration, the synchronous PDMM activated all the nodes for parameter-updating. As for the asynchronous PDMM, the nodes were activated sequentially by following the mapping $i=\textrm{mod}(k,m)+1$, where the iteration $k\geq 0$ (See Subsection~\ref{subsec:rate_asyn}). As a result, after every segment of $m=100$ iterations, all the nodes were activated once. In the experiment, we counted the number of iterations for the synchronous PDMM and the number of segments (each segment consists of $m$ iterations) for the asynchronous PDMM.

For each parameter-setting, we initialized $(\hat{{x}}_i^{0},\hat{\boldsymbol{\lambda}}_{i}^{0})=(t_i,\boldsymbol{0})$ for every $i\in \mathcal{V}$. The algorithm stops when the squared error is below $10^{-4}$.

Fig.~\ref{fig:per_param} displays the numbers of iterations (or segments) of PDMM under different parameter-settings. Each $\circ$ or {\tiny $\square$} symbol represents a particular setting for $(\gamma_p,\gamma_d)$. The settings denoted by {\tiny $\square$} are for the case that $\gamma_p\gamma_d<1$ while the ones by $\circ$ are for the case that $\gamma_p\gamma_d\geq 1$.

\begin{figure}[tb]
\centering
\begin{footnotesize}
  \includegraphics[width=80mm]{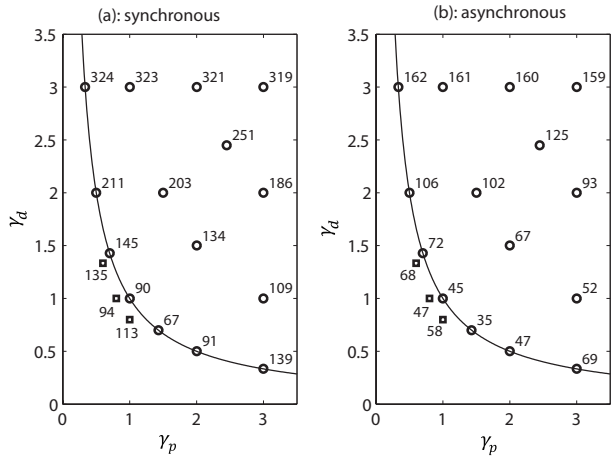}
\end{footnotesize}
\caption{\small Performance of PDMM for different parameter settings. Each value in subplot~$(a)$ represents the number of iterations required for the synchronous PDMM. On the other hand, each value in subplot~(b) represents the number of segments of iterations for the asynchronous PDMM, where each segment consists of $100$ iterations.  The convex curve in each subplot corresponds to $\gamma_p\gamma_d=1$.}
\label{fig:per_param}
\end{figure}

It is seen from the figure that large $\gamma_p$ or $\gamma_d$ can only make the algorithm converge slowly. The optimal parameter-setting that leads to the fastest convergence lies on the curve $\gamma_d\gamma_p=1$ for both the synchronous and the asynchronous updating schemes. Further, it appears that the two optimal settings for the two updating schemes are in a neighborhood.

Finally, we note that the settings denoted by {\tiny $\square$} correspond to the situation that $\gamma_p\gamma_d<1$. The experiment for those settings demonstrates that Condition~\ref{con:G} is only sufficient for algorithmic convergence. We also tested the setting $\gamma_p=\gamma_d=0.5$. We found that the above setting led to  divergence for both synchronous and synchronous schemes. This phenomenon suggests that $\gamma_p$ and $\gamma_d$ cannot be chosen arbitrarily in practice.

\subsubsection{performance with transmission failure}
\label{subsub:loss}
In this experiment, we studied how transmission failure affects the performance of PDMM given the fact that no convergence guaranty is derived at the moment. As discussed in Subsection~\ref{subsec:node_comp_trans}, we could not use broadcast transmission in the case of transmission loss. Instead, each activated node $i$ has to perform point-to-point transmission for $\hat{\lambda}_{i|j}$ from node $i$ to node $j\in \mathcal{N}_i$.

Due to transmission failure, PDMM was initialized differently from the first experiment. Each time the algorithm was tested, the initial estimate $(\hat{{\boldsymbol{x}}}^{0}, \hat{\boldsymbol{\lambda}}^{0})$ was set as
\begin{align}
(\hat{{\boldsymbol{x}}}^{0}, \hat{\boldsymbol{\lambda}}^{0})=(\boldsymbol{0},\boldsymbol{0}). \label{equ:per_loss_init}
\end{align}
The above initialization guarantees that every node in the graph has access to the initial estimates of neighboring nodes without transmission.

Fig.~\ref{fig:per_loss} demonstrates the performance of PDMM under three transmission losses: 0\%, 20\% and 40\%. Subplot~(a) and (b) are for the asynchronous and synchronous schemes, respectively. Each curve in the two subplots was obtained by averaging over 100 simulations to mitigate the effect of random transmission losses. It is seen that transmission failure only slows down the convergence speed of the algorithm. The above property is highly desirable in real applications because transmission losses might be inevitable in some networks (e.g., see \cite{Zhao03WSN} for investigation of packet-loss over wireless sensor networks in different environments).
 
Finally, it is observed that for each transmission-loss in subplot~(a), the error goes up in the first few hundred of iterations before deceasing. This may because of the special initialization (\ref{equ:per_loss_init}). We have tested the initialization $\{\hat{x}_i^{0}=t_i\}$ for 0\% transmission loss, where the MSE decreases along with the iterations monotonically. 

\begin{figure}[tb]
\centering
\begin{footnotesize}
  \includegraphics[width=80mm]{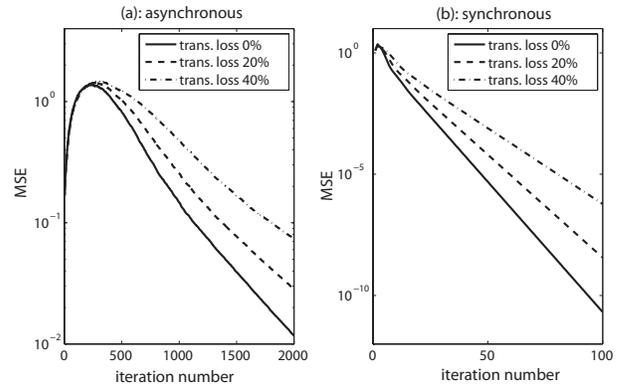}
\end{footnotesize}
\caption{\small Performance of synchronous/asynchronous PDMM under different transmission losses (\%). } \label{fig:per_loss}
\end{figure}

\subsubsection{performance comparison}
In this experiment, we investigated the convergence speeds of four algorithms under the condition of no transmission failure. Besides PDMM, we also implemented the broadcast-based algorithm in \cite{Iutzeler13gossipAlg} (referred to as \emph{broadcast}), the randomized gossip algorithm in \cite{Boyd06gossip} (referred to as \emph{gossip}) and ADMM. Unlike PDMM and ADMM that can work either synchronously or asynchronously, both \emph{broadcast} and \emph{gossip} algorithms can only work asynchronously. While \emph{broadcast} algorithm randomly activates one node per iteration, \emph{gossip} algorithm randomly activates one edge per iteration for parameter-updating.

Similarly to the first experiment, we also evaluated PDMM for both the synchronous and asynchronous schemes. For the asynchronous scheme, we tested all the four algorithms introduced above while for the synchronous scheme, we focused on PDMM and ADMM. The implementation of the synchronous/asynchronous ADMM follows from \cite{Boyd11ADMM} and \cite{Wei13ADMM}, respectively. The asynchronous ADMM \cite{Wei13ADMM} is similar to the \emph{gossip} algorithm in the sense that both algorithms activates one edge per iteration. 

We note that the asynchronous ADMM essentially activates two neighboring nodes per iteration. To make a fair comparison between PDMM and ADMM, we implemented two versions of PDMM for the asynchronous scheme. The first version follows Subsection~\ref{subsec:asyn_updating} where each iteration randomly activates one node as the \emph{gossip} algorithm, referred to as \emph{one-node PDMM}. The second version of PDMM randomly activates two neighboring nodes per iteration as the \emph{broadcast} algorithm, referred to as \emph{two-node PDMM}.

Both PDMM and ADMM have some parameters to be specified. To simplify the implementation, we let $\gamma_p=\gamma_d=1$ in PDMM (which is not the optimal setting from Fig.~\ref{fig:per_param}). 
Similarly,  we set the parameter in ADMM to be 1.

In the experiment, the \emph{gossip} and \emph{broadcast} algorithms were initialized according to \cite{Boyd06gossip} and \cite{Iutzeler13gossipAlg}, respectively. The initialization for PDMM was the same as in the first experiment. The estimates of ADMM were initialized similarly as for PDMM.

\begin{figure}[tb]
\centering
\begin{footnotesize}
  \includegraphics[width=80mm]{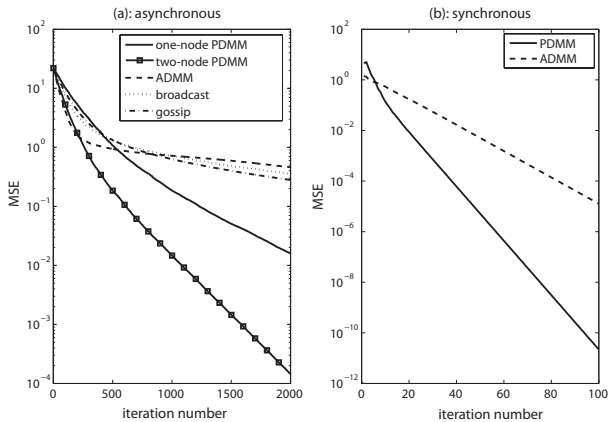}
\end{footnotesize}
\caption{\small  Performance comparison under perfect channel. The two curves in subplot~(b) at iteration 1 have a noticeable gap compared to subplot~(a). This is because under the synchronous scheme, all the parameters of each method are updated per iteration, leading to a relatively big performance difference in the beginning. } \label{fig:per_perfect}
\end{figure}
\begin{table}[t]
\centering
\begin{tabular}{|c|c|c|c|c|c|c|c|}
\hline
&\hspace{-3.5mm} {{\scriptsize $\begin{array}{c}\textrm{one-node} \\ \textrm{PDMM}\end{array}$}}\hspace{-3mm}
& \hspace{-3mm}{\scriptsize $\begin{array}{c}\textrm{two-node} \\ \textrm{PDMM}\end{array}$}  \hspace{-4mm}
& \hspace{-1.8mm}{\scriptsize ADMM} \hspace{-3mm} 
& \hspace{-1.5mm}\scriptsize{broadcast} \hspace{-1.8mm}
& \hspace{-1.8mm}\scriptsize{gossip}   \hspace{-2mm}
& \hspace{-3.5mm} {{\scriptsize $\begin{array}{c} \textrm{PDMM} \\ \textrm{(syn)}\end{array}$}}\hspace{-3mm} 
& \hspace{-3.5mm} {{\scriptsize $\begin{array}{c} \textrm{ADMM} \\ \textrm{(syn)}\end{array}$}}\hspace{-3mm}  \\
\hline  
\hspace{-1.5mm}\scriptsize{ave. ($\mu s$)} \hspace{-1.5mm} & \footnotesize{5.46} & \footnotesize{8.92}& \footnotesize{6.54}& \footnotesize{2.10} & 
\footnotesize{0.24} & \footnotesize{380} & \footnotesize{384} \\ 
\hline
\hspace{-1.5mm}\scriptsize{std ($10^{-6}$)}  \hspace{-1.5mm} & \footnotesize{5.04} & \footnotesize{8.58} & \footnotesize{8.09} & \footnotesize{4.55} & \footnotesize{1.73} &\footnotesize{216} & \footnotesize{285} \\ 
\hline
\end{tabular}
\caption{\small Average execution times (per iteration) and their standard deviations for the four methods. } 
\label{tab:time} \vspace{-5mm}
\end{table}

Fig.~\ref{fig:per_perfect} displays the MSE trajectories for the four methods while Table~\ref{tab:time} lists the average execution times (per iteration) and their standard deviations. Similarly to the second experiment, the performance of each method for the asynchronous scheme was obtained by averaging over 100 simulations to mitigate the effect of randomness introduced in node or edge-activation.
We now focus on the asynchronous scheme. It is seen that the \emph{two-node PDMM} converges the fastest in terms of number of iterations while the \emph{gossip} algorithm requires the least execution time on average.  
The above results suggest that for applications where signal transmission is more expensive than local computation (w.r.t. energy consumption), PDMM might be a good candidate as it may save number of iterations. 

Fig.~\ref{fig:per_perfect}~(b) demonstrates the MSE performance of PDMM and ADMM for the synchronous scheme.   Both algorithms appear to have linear convergence rates. This may be because the objective functions in (\ref{equ:avePrim}) are strongly convex and have gradients which are Lipschitz continuous. It is seen from Table~\ref{tab:time} that both methods take roughly the same execution time. By combining the above results, we conclude that under synchronous scheme, PDMM converges faster than ADMM w.r.t. the execution time, which may be due to the fact that PDMM avoids the auxiliary variable $\boldsymbol{z}$ used in ADMM.

\section{Conclusion}
\label{sec:conclusion}
In this paper, we have proposed PDMM for iterative optimization over a general graph. The augmented primal-dual Lagrangian function is constructed of which a saddle point provides an optimal solution of the original problem, which leads to the design of PDMM. PDMM performs broadcast transmission under perfect channel and point-to-point transmission under non-perfect channel. We have shown that both the synchronous and asynchronous PDMMs possess a convergence rate of $\mathcal{O}(1/K)$ for general closed, proper and convex functions defined over the graph. As an example, we have applied PDMM for distributed averaging, through which properties of PDMM such as proper parameter-selection and resilience against transmission failure are further investigated. 

We note that PDMM is natural when performing node-oriented optimization over a graph as compared to ADMM which involves computing the edge variable $\boldsymbol{z}$ introduced in (\ref{equ:optProTwo}). A few applications in \cite{Heming15Thesis}, \cite{xiaoqiang16BiADMM} and \cite{Sherson16LCMV_conf}  suggest that PDMM is practically promising. While convergence properties of ADMM under different conditions (e.g., strong convexity and/or the gradients being Lipschitz continuous) are well understood, the convergence properties of PDMM for those conditions remain to be discovered.      

\appendices

\section{Proof for Lemma \ref{lemma:gen_syn}}
\label{appendix:lemma_gen_syn}

Before presenting the proof, we first introduce a basic inequality, which is described in a lemma below:

\begin{lemma}
Let $f_1(\boldsymbol{x})$ and $f_2(\boldsymbol{x})$ be two arbitrary closed, proper and convex functions. $\boldsymbol{x}^{\star}$ minimizes the sum of the two functions, i.e., $\boldsymbol{x}^{\star}=\arg\min_{\boldsymbol{x}}(f_1(\boldsymbol{x})+f_2(\boldsymbol{x}))$. Then, there is
\begin{align}
f_1(\boldsymbol{x})-f_1(\boldsymbol{x}^{\star})\geq (\boldsymbol{x}^{\star}-\boldsymbol{x})^T\boldsymbol{r}(\boldsymbol{x}^{\star})\quad \forall \boldsymbol{x}, \label{equ:basic_VI}
\end{align}
where $\boldsymbol{r}(\boldsymbol{x}^{\star})\in \partial_{\boldsymbol{x}}f_2(\boldsymbol{x}^{\star})$.
\label{lemma:basic_VI}
\end{lemma}

The above inequality is wildly exploited for the convergence analysis of ADMM and its variants \cite{Wang12OADM,Deng16ADMM,Boyd11ADMM}. We will also use the inequality in our proof.

Applying (\ref{equ:basic_VI}) to the updating equations (\ref{equ:x_updateSyn})-(\ref{equ:lambda_updateSyn}) for $(\hat{\boldsymbol{x}}^{k+1},\hat{\boldsymbol{\lambda}}^{k+1})$, we obtain a set of inequalities for all $ (\boldsymbol{x},\boldsymbol{\lambda})\in (\mathbb{R}^{\sum n_i},\mathbb{R}^{2\sum n_{ij}})$ as
\begin{align}
&\hspace{-2.5mm} \sum_{j\in \mathcal{N}_i}\hspace{-1mm}\left[\boldsymbol{P}_{d,ij}(\hat{\boldsymbol{\lambda}}_{j|i}^{k}\hspace{-0.5mm}-\hspace{-0.5mm}\hat{\boldsymbol{\lambda}}_{i|j}^{k+1})\hspace{-0.5mm}+\hspace{-0.5mm}\boldsymbol{c}_{ij}\hspace{-0.5mm}-\hspace{-0.5mm}\boldsymbol{A}_{j i}\hat{\boldsymbol{x}}_j^{k}\right]^T\hspace{-1.2mm}(\boldsymbol{\lambda}_{i|j}\hspace{-0.5mm}-\hspace{-0.5mm}\hat{\boldsymbol{\lambda}}_{i|j}^{k+1})  \nonumber \\
&\hspace{-2.5mm} \leq \hspace{-1mm} f_i^{\ast}(\boldsymbol{A}_{i}^T\boldsymbol{\lambda}_{i})\hspace{-0.5mm}-\hspace{-0.5mm} f_i^{\ast}(\boldsymbol{A}_{i}^T\hat{\boldsymbol{\lambda}}_{i}^{k+1})\hspace{26mm} \forall i\in \mathcal{V}\label{equ:optConSyn2}\\
&\hspace{-2.5mm}\sum_{j\in \mathcal{N}_i}\hspace{-1.2mm}\Big[\Big(\boldsymbol{P}_{p,ij}(\boldsymbol{c}_{ij}\hspace{-0.4mm}-\hspace{-0.4mm}\boldsymbol{A}_{i j}\boldsymbol{x}_i^{k+1}\hspace{-0.4mm}-\hspace{-0.4mm}\boldsymbol{A}_{j i}\hat{\boldsymbol{x}}_j^{k})\hspace{-0.4mm}+\hspace{-0.4mm}\hat{\boldsymbol{\lambda}}_{j|i}^{k}\Big)^T\hspace{-1.5mm} \boldsymbol{A}_{i j} \nonumber\\
&\hspace{8mm}\cdot (\boldsymbol{x}_i-\hat{\boldsymbol{x}}_i^{k+1} )\Big]\leq  f_i(\boldsymbol{x}_i)-f_i(\hat{\boldsymbol{x}}_i^{k+1})\hspace{4mm} \forall i\in \mathcal{V}. \label{equ:optConSyn1}
\end{align}
Adding (\ref{equ:optConSyn2})-(\ref{equ:optConSyn1}) over all $i\in \mathcal{V}$, and substituting $(\boldsymbol{x},\boldsymbol{\lambda})=(\boldsymbol{x}^{\star},\boldsymbol{\lambda}^{\star})$, the saddle point of $L_{\mathcal{P}}$, yields
\begin{align}
&\hspace{-2mm}\sum_{i\in \mathcal{V}}\sum_{j\in \mathcal{N}_i}\hspace{-1.2mm}\Big[(\hat{\boldsymbol{\lambda}}_{i|j}^{k+1}\hspace{-0.5mm}-\hspace{-0.5mm}\boldsymbol{\lambda}_{i|j}^{\star})^T\Big(\boldsymbol{A}_{j i}\hat{\boldsymbol{x}}_j^{k+1}\hspace{-0.5mm}-\hspace{-0.5mm}\frac{\boldsymbol{c}_{ij}}{2}\Big)\hspace{-0.5mm}-\hspace{-0.5mm}(\hat{\boldsymbol{x}}_i^{k+1}\hspace{-0.5mm}-\hspace{-0.5mm}\boldsymbol{x}_i^{\star})^T\nonumber\\
&\hspace{14mm}\cdot\boldsymbol{A}_{i j}^T\hat{\boldsymbol{\lambda}}_{j|i}^{k+1}\Big] +p(\hat{\boldsymbol{x}}^{k+1},\hat{\boldsymbol{\lambda}}^{k+1})-p(\boldsymbol{x}^{\star},\boldsymbol{\lambda}^{\star})\nonumber \\
&\leq \hspace{-0.5mm} \sum_{i\in \mathcal{V}}\hspace{-0.5mm}\sum_{j\in \mathcal{N}_i}\hspace{-1mm} \Big[\Big(\boldsymbol{P}_{p,ij}(\boldsymbol{c}_{ij}\hspace{-0.4mm}-\hspace{-0.4mm}\boldsymbol{A}_{i j}\boldsymbol{x}_i^{k+1}\hspace{-0.4mm}-\hspace{-0.4mm}\boldsymbol{A}_{j i}\hat{\boldsymbol{x}}_j^{k})+\hspace{-0.4mm}\hat{\boldsymbol{\lambda}}_{j|i}^{k} \nonumber\\
&\hspace{20mm}-\hspace{-0.5mm}\hat{\boldsymbol{\lambda}}_{j|i}^{k+1}\Big)^T\hspace{-1mm}\boldsymbol{A}_{i j}(\hat{\boldsymbol{x}}_i^{k+1}\hspace{-0.5mm}-\hspace{-0.5mm}\boldsymbol{x}_i^{\star})\hspace{-0.5mm}+\hspace{-0.5mm}(\hat{\boldsymbol{\lambda}}_{i|j}^{k+1}\hspace{-0.6mm}-\hspace{-0.6mm}\boldsymbol{\lambda}_{i|j}^{\star})^T\nonumber\\
&\hspace{15mm}\cdot\hspace{-0.5mm}\left(\hspace{-0.3mm}\boldsymbol{P}_{d,ij}(\hat{\boldsymbol{\lambda}}_{j|i}^{k}\hspace{-0.5mm}-\hspace{-0.5mm}\hat{\boldsymbol{\lambda}}_{i|j}^{k+1})\hspace{-0.5mm}+\hspace{-0.5mm}\boldsymbol{A}_{j i}(\hat{\boldsymbol{x}}_{j}^{k+1}\hspace{-0.5mm}-\hspace{-0.5mm}\hat{\boldsymbol{x}}_j^{k})\hspace{-0.3mm}\right)\hspace{-0.7mm}\Big] \nonumber \\
&= \hspace{-0.5mm} \sum_{i\in \mathcal{V}}\hspace{-0.5mm}\sum_{j\in \mathcal{N}_i}\hspace{-1mm} \Big[\Big(\boldsymbol{P}_{p,ij}\boldsymbol{A}_{j i}(\hat{\boldsymbol{x}}_j^{k+1}\hspace{-0.4mm}-\hspace{-0.4mm}\hat{\boldsymbol{x}}_j^{k})+\hspace{-0.4mm}\hat{\boldsymbol{\lambda}}_{j|i}^{k}\hspace{-0.4mm}-\hspace{-0.4mm}\hat{\boldsymbol{\lambda}}_{j|i}^{k+1}\Big)^T \nonumber\\
&\hspace{20mm}\cdot\boldsymbol{A}_{i j}(\hat{\boldsymbol{x}}_i^{k+1}-\boldsymbol{x}_i^{\star})+(\hat{\boldsymbol{\lambda}}_{i|j}^{k+1}-\boldsymbol{\lambda}_{i|j}^{\star})^T\nonumber\\
&\hspace{19mm}\cdot\hspace{-0.5mm}\left(\boldsymbol{P}_{d,ij}(\hat{\boldsymbol{\lambda}}_{j|i}^{k}\hspace{-0.5mm}-\hspace{-0.5mm}\hat{\boldsymbol{\lambda}}_{j|i}^{k+1})\hspace{-0.5mm}+\hspace{-0.5mm}\boldsymbol{A}_{j i}(\hat{\boldsymbol{x}}_{j}^{k+1}\hspace{-0.5mm}-\hspace{-0.5mm}\hat{\boldsymbol{x}}_j^{k})\right)\Big]\nonumber\\
&\hspace{3mm}-\hspace{-1.5mm}\sum_{(i,j)\in \mathcal{E}}\hspace{-1.0mm} \Big(\|\boldsymbol{c}_{ij}\hspace{-0.4mm}-\hspace{-0.4mm}\boldsymbol{A}_{i j}\boldsymbol{x}_i^{k+1}\hspace{-0.4mm}-\hspace{-0.4mm}\boldsymbol{A}_{j i}\hat{\boldsymbol{x}}_j^{k+1} \|_{\boldsymbol{P}_{p,ij}}^2\nonumber\\
&\hspace{16mm}+\|\hat{\boldsymbol{\lambda}}_{i|j}^{k+1}-\hat{\boldsymbol{\lambda}}_{j|i}^{k+1}\|_{\boldsymbol{P}_{d,ij}}^2\Big),
\label{equ:optCon_gen_syn1} 
\end{align}
where the last equality follows from the two optimality conditions (\ref{equ:L_G_opt3})-(\ref{equ:L_G_opt4}).

To further simplify (\ref{equ:optCon_gen_syn1}), one can first insert the alternative expression (\ref{equ:G_dp}) for every $\boldsymbol{P}_{d,ij}$ into (\ref{equ:optCon_gen_syn1}). After that, the expression (\ref{equ:optCon_syn}) can be obtained by simplifying the new expression  using (\ref{equ:L_G_opt3})-(\ref{equ:L_G_opt4}) and the following identity
\begin{align}
&(\boldsymbol{y}_1-\boldsymbol{y}_2)^T(\boldsymbol{y}_3-\boldsymbol{y}_4)\nonumber\\
&\equiv\hspace{-0.5mm}\frac{1}{2}(\|\boldsymbol{y}_1\hspace{-0.5mm}+\hspace{-0.5mm}\boldsymbol{y}_3\|^2\hspace{-0.5mm}-\hspace{-0.5mm}\|\boldsymbol{y}_1\hspace{-0.5mm}+\hspace{-0.5mm}\boldsymbol{y}_4\|^2\hspace{-0.5mm}-\hspace{-0.5mm}\|\boldsymbol{y}_2\hspace{-0.5mm}+\hspace{-0.5mm}\boldsymbol{y}_3\|^2\hspace{-0.5mm}+\hspace{-0.5mm}\|\boldsymbol{y}_2\hspace{-0.5mm}+\hspace{-0.5mm}\boldsymbol{y}_4\|^2).\nonumber
\end{align}

\section{Proof of Lemma~\ref{lemma:gen_asyn}}
\label{appendix:lemma_gen_asyn}

The basic idea for the proof is similar to that for Lemma~\ref{lemma:gen_syn} as presented in Appendix~\ref{appendix:lemma_gen_syn}.
However, since asynchronous PDMM activates one node $i\in \mathcal{V}$ per iteration, it is difficult to tell which neighbors of $i$ have been recently activated and which have not yet. The above difficulty requires careful treatment in the convergence analysis. We sketch the proof in the following for reference.


We focus on the parameter-updating for a particular segment of iterations $k\in \{ml, ml+1,\ldots,ml+m-1\}$, where $l\geq 0$. For simplicity, we denote the activated node $i$ at iteration $k$ as $i(k)$. To start with, we apply (\ref{equ:basic_VI}) to the updating equation (\ref{equ:x_lambda_updateAsyn1}) for the estimate $(\hat{\boldsymbol{x}}_{i(k)}^{k+1},\hat{\boldsymbol{\lambda}}_{i(k)}^{k+1})$ of node $i(k)$. In order to do so, we first have to consider the estimates of its neighbors. It may happen that some neighbors have already been activated within the segment while others are still waiting to be activated. If a neighbor $j\in \mathcal{N}_{i(k)}$ is still waiting, we then have $(\hat{\boldsymbol{x}}_j^k,\hat{\boldsymbol{\lambda}}_j^k)=(\hat{\boldsymbol{x}}_j^{lm},\hat{\boldsymbol{\lambda}}_j^{lm})$. Conversely, if a neighbor $j\in \mathcal{N}_{i(k)}$ has already been activated, we then have $(\hat{\boldsymbol{x}}_j^k,\hat{\boldsymbol{\lambda}}_j^k)=(\hat{\boldsymbol{x}}_j^{(l+1)m},\hat{\boldsymbol{\lambda}}_j^{(l+1)m})$. From Lemma~\ref{lemma:act_node_order}, it is clear that if $j<i(k)$ (or $j>i(k)$), then the neighbor $j$ has been activated (not yet activated). For simplicity, we use a function $s(k,j)$ to denote the value $lm$ or $(l+1)m$ for a neighbor $j\in \mathcal{N}_{i(k)}$ at iteration $k$
\begin{align}
s(k,j)=\left\{\begin{array}{ll}lm       & j>i(k)\\
                                (l+1)m  & j<i(k)  \end{array}\right..\label{equ:sFun}
\end{align} As for the activated node $i(k)$, we have $(\hat{\boldsymbol{x}}_{i(k)}^{k+1},\hat{\boldsymbol{\lambda}}_{i(k)}^{k+1})=(\hat{\boldsymbol{x}}_{i(k)}^{(l+1)m},\hat{\boldsymbol{\lambda}}_{i(k)}^{(l+1)m})$. As a result, the two inequalities for $\hat{\boldsymbol{x}}_{i(k)}^{k+1}$ and $\hat{\boldsymbol{\lambda}}_{i(k)}^{k+1}$ are given by
\begin{align}
&\hspace{-1mm} \sum_{j\in \mathcal{N}_{i(k)}}\hspace{-1mm}\Big[\boldsymbol{P}_{d,i(k)j}(\hat{\boldsymbol{\lambda}}_{j|i(k)}^{s(k,j)}\hspace{-0.5mm}-\hspace{-0.5mm}\hat{\boldsymbol{\lambda}}_{i(k)|j}^{(l+1)m})\hspace{-0.5mm}-\hspace{-0.5mm}\boldsymbol{A}_{j i(k)}\hat{\boldsymbol{x}}_j^{s(k,j)}\nonumber \\
&\hspace{14mm}+\hspace{-0.5mm}\boldsymbol{c}_{i(k)j}\Big]^T\hspace{-1mm}\left(\boldsymbol{\lambda}_{i(k)|j}\hspace{-0.5mm}-\hspace{-0.5mm}\hat{\boldsymbol{\lambda}}_{i(k)|j}^{(l+1)m}\right)  \nonumber \\
&\hspace{-1mm} \leq \hspace{-1mm} f_{i(k)}^{\ast}\hspace{-0.8mm}\left(\boldsymbol{A}_{i(k)}^T\boldsymbol{\lambda}_{i(k)}\right)\hspace{-0.5mm}-\hspace{-0.5mm} f_{i(k)}^{\ast}\hspace{-0.8mm}\left(\boldsymbol{A}_{i(k)}^T\hat{\boldsymbol{\lambda}}_{i(k)}^{(l+1)m}\right) \label{equ:optConAsyn1}\\
&\hspace{-1mm}\sum_{j\in \mathcal{N}_{i(k)}}\hspace{-1.2mm}\Big[\boldsymbol{P}_{p,i(k)j}\Big(-\hspace{-0.4mm}\boldsymbol{A}_{i(k) j}\boldsymbol{x}_{i(k)}^{(l+1)m}\hspace{-0.4mm}-\hspace{-0.4mm}\boldsymbol{A}_{j i(k)}\hat{\boldsymbol{x}}_j^{s(k,j)} \nonumber\\
&\hspace{8mm} +\boldsymbol{c}_{i(k)j}\hspace{-0.4mm}\Big)\hspace{-0.4mm}+\hspace{-0.4mm}\hat{\boldsymbol{\lambda}}_{j|i(k)}^{s(k,j)}\Big]^T \boldsymbol{A}_{i(k) j}\left(\boldsymbol{x}_{i(k)}-\hat{\boldsymbol{x}}_{i(k)}^{(l+1)m} \right)\nonumber \\
&\hspace{-1mm} \leq  f_{i(k)}\hspace{-0.8mm}\left(\boldsymbol{x}_{i(k)}\right)-f_{i(k)}\hspace{-0.8mm}\left(\hat{\boldsymbol{x}}_{i(k)}^{(l+1)m}\right), \label{equ:optConAsyn2}
\end{align}
where $lm\leq k<(l+1)m$.

Next adding (\ref{equ:optConAsyn1})-(\ref{equ:optConAsyn2}) over all $lm\leq k < (l+1)m$ and substituting $(\boldsymbol{x},\boldsymbol{\lambda})=(\boldsymbol{x}^{\star},\boldsymbol{\lambda}^{\star})$ yields
\begin{align}
&\hspace{-0mm}\sum_{i\in \mathcal{V}}\hspace{-0.8mm}\sum_{j\in \mathcal{N}_i}\hspace{-1.2mm}\Big[\hspace{-1.5mm}\left(\hspace{-0.6mm}\hat{\boldsymbol{\lambda}}_{i|j}^{(l\hspace{-0.3mm}+\hspace{-0.4mm}1)m}\hspace{-0.9mm}-\hspace{-0.9mm}\boldsymbol{\lambda}_{i|j}^{\star}\hspace{-0.6mm}\right)^{\hspace{-0.6mm}T}\hspace{-2mm}\Big(\hspace{-0.5mm}\boldsymbol{A}_{j i}\hat{\boldsymbol{x}}_j^{(l\hspace{-0.4mm}+\hspace{-0.4mm}1)m}\hspace{-0.7mm}-\hspace{-0.7mm}\frac{\boldsymbol{c}_{ij}}{2}\hspace{-0.5mm}\Big)\hspace{-1mm}-\hspace{-1mm}\Big(\hspace{-0.5mm}\hat{\boldsymbol{x}}_i^{(l\hspace{-0.3mm}+\hspace{-0.3mm}1)m}\hspace{-0.7mm}-\hspace{-0.7mm}\boldsymbol{x}_i^{\star}\hspace{-0.5mm}\Big)^{\hspace{-0.6mm}T}\hspace{-0.5mm}\nonumber\\
&\hspace{4mm}\cdot\boldsymbol{A}_{i j}^T\hat{\boldsymbol{\lambda}}_{j|i}^{(l+1)m}\Big]\hspace{-0.5mm}+\hspace{-0.5mm}p\left(\hat{\boldsymbol{x}}^{(l+1)m},\hat{\boldsymbol{\lambda}}^{(l+1)m}\right) -p(\boldsymbol{x}^{\star},\boldsymbol{\lambda}^{\star})\nonumber \\
&\leq \sum_{k=lm}^{(l\hspace{-0.3mm}+\hspace{-0.3mm}1)m\hspace{-0.3mm}-\hspace{-0.3mm}1}\sum_{j\in \mathcal{N}_{i(k)}}\Bigg[\Big[\boldsymbol{P}_{d,i(k)j}\left(\hat{\boldsymbol{\lambda}}_{j|i(k)}^{s(k,j)}-\hat{\boldsymbol{\lambda}}_{i(k)|j}^{(l+1)m}\right)\nonumber\\
&\hspace{8mm}+\boldsymbol{A}_{j i(k)}
\left(\hat{\boldsymbol{x}}_j^{(l+1)m}\hspace{-0.4mm}-\hspace{-0.4mm}\hat{\boldsymbol{x}}_j^{s(k,j)}\right)\nonumber\Big]^T\hspace{-0.8mm}
\left(\hat{\boldsymbol{\lambda}}_{i(k)|j}^{(l+1)m}\hspace{-0.4mm}-\hspace{-0.4mm}\boldsymbol{\lambda}_{i(k)|j}^{\star}\right)\nonumber \\
&\hspace{4mm}+\Big[\boldsymbol{P}_{p,i(k)j}\Big(\boldsymbol{c}_{i(k)j}\hspace{-0.4mm}-\hspace{-0.4mm}\boldsymbol{A}_{i(k) j}
\boldsymbol{x}_{i(k)}^{(l+1)m}\hspace{-0.4mm}-\hspace{-0.4mm}\boldsymbol{A}_{j i(k)}
\hat{\boldsymbol{x}}_j^{s(k,j)}\Big) \nonumber\\
&\hspace{8mm} \hspace{-0.4mm}+\hspace{-0.4mm}\hat{\boldsymbol{\lambda}}_{j|i(k)}^{s(k,j)}\hspace{-0.4 mm}-\hspace{-0.4mm}\hat{\boldsymbol{\lambda}}_{j|i(k)}^{(l+1)m}\Big]^T \boldsymbol{A}_{i(k) j}
\left(\hat{\boldsymbol{x}}_{i(k)}^{(l+1)m} -\boldsymbol{x}_{i(k)}^{\star} \right) \Bigg]\nonumber \\
&= \hspace{-2mm}\sum_{k=lm}^{(l\hspace{-0.3mm}+\hspace{-0.3mm}1)m\hspace{-0.3mm}-\hspace{-0.3mm}1}\hspace{-2mm}\sum_{j\in \mathcal{N}_{i(k)}}\hspace{-1mm}g(k,i(k),j)\hspace{-0.6mm}-\hspace{-0.6mm} \sum_{(i,j)\in \mathcal{E}}\hspace{-2mm}\Big(\hspace{-0.2mm}\left\|\hat{\boldsymbol{\lambda}}_{i|j}^{(l\hspace{-0.3mm}+\hspace{-0.3mm}1)m}\hspace{-0.6mm}-\hspace{-0.6mm}\hat{\boldsymbol{\lambda}}_{j|i}^{(l\hspace{-0.3mm}+\hspace{-0.3mm}1)m} \right\|_{\boldsymbol{P}_{d,ij}}^2 \nonumber \\
&\hspace{10mm}+ \left\|\boldsymbol{c}_{ij}-\boldsymbol{A}_{i j}
\hat{\boldsymbol{x}}_i^{(l+1)m}-\boldsymbol{A}_{j i}
\hat{\boldsymbol{x}}_j^{(l+1)m}\right\|_{\boldsymbol{P}_{p,ij}}^2 \hspace{-0.2mm}\Big),\label{equ:optCon_gen_asyn1}
\end{align}
where the function $g(k,i(k),j)$ is defined as
\begin{align}
&g(k,i(k),j)\nonumber\\
&=\Big[\boldsymbol{P}_{d,i(k)j}\left(\hat{\boldsymbol{\lambda}}_{j|i(k)}^{s(k,j)}-\hat{\boldsymbol{\lambda}}_{j|i(k)}^{(l+1)m}\right)\nonumber\\
&\hspace{8mm}+\hspace{-0.5mm}\boldsymbol{A}_{j i(k)}
\left(\hat{\boldsymbol{x}}_j^{(l+1)m}\hspace{-0.4mm}-\hspace{-0.4mm}\hat{\boldsymbol{x}}_j^{s(k,j)}\right)\nonumber\Big]^T\hspace{-0.8mm}
\left(\hat{\boldsymbol{\lambda}}_{i(k)|j}^{(l+1)m}\hspace{-0.4mm}-\hspace{-0.4mm}\boldsymbol{\lambda}_{i(k)|j}^{\star}\right)\nonumber\\
&\hspace{4mm}+\Big[\boldsymbol{P}_{p,i(k)j}\boldsymbol{A}_{j i(k)}
\Big(\hat{\boldsymbol{x}}_j^{(l+1)m}\hspace{-0.4mm}-\hspace{-0.4mm}\hat{\boldsymbol{x}}_j^{s(k,j)}\Big) \nonumber\\
&\hspace{8mm} \hspace{-0.4mm}+\hspace{-0.4mm}\hat{\boldsymbol{\lambda}}_{j|i(k)}^{s(k,j)}\hspace{-0.4 mm}-\hspace{-0.4mm}\hat{\boldsymbol{\lambda}}_{j|i(k)}^{(l+1)m}\Big]^T \boldsymbol{A}_{i(k) j}
\left(\hat{\boldsymbol{x}}_{i(k)}^{(l+1)m} -\boldsymbol{x}_{i(k)}^{\star} \right),\nonumber
\end{align}
where $lm\leq k<(l+1)m$ and $j\in \mathcal{N}_{i(k)}$.

Now we are in a position to analyze the right hand side of (\ref{equ:optCon_gen_asyn1}). By using the fact that each node $i$ has $|\mathcal{N}_i|$ different functions $g(k,i(k),j)$, we can conclude that each edge $(u,v)\in \mathcal{E}$ is associated with two functions $g(k_1,u(k_1),v)$ and $g(k_2,v(k_2),u)$, where iteration $k_1$ and $k_2$ activate $u$ and $v$, respectively. From (\ref{equ:optCon_gen_asyn1}), it is clear that each edge $(u,v)$ is also associated with the other two functions $\|\boldsymbol{c}_{uv}-\boldsymbol{A}_{u v}
\hat{\boldsymbol{x}}_u^{(l+1)m}-\boldsymbol{A}_{v u}\hat{\boldsymbol{x}}_v^{(l+1)m}\|_{\boldsymbol{P}_{p,uv}}^2$ and $\|\hat{\boldsymbol{\lambda}}_{v|u}^{(l+1)m}-\hat{\boldsymbol{\lambda}}_{u|v}^{(l+1)m}\|_{\boldsymbol{P}_{d,uv}}^2$.
We show in the following that the combination of the above four functions for every edge $(u,v)\in \mathcal{E}$ is independent of $k_1$ and $k_2$.
In order to do so, we assume $k_1<k_2$ (or equivalently, $u<v$ from Lemma~\ref{lemma:act_node_order}). From (\ref{equ:sFun}), we know that $s(k_1,v)=lm$ and $s(k_2,u)=(l+1)m$. Based on the above information, the four functions for $(u,v)\in \mathcal{E}$ can be simplified as 
\begin{align}
&\hspace{0mm}g(k_1,\hspace{-0.2mm}u(k_1),\hspace{-0.2mm}v)\hspace{-0.5mm}+\hspace{-0.5mm}g(k_2,\hspace{-0.2mm}v(k_2),\hspace{-0.2mm}u)\hspace{1mm}-\hspace{-0.7mm}\|\hat{\boldsymbol{\lambda}}_{v|u}^{(l+1)m}\hspace{-0.7mm}-\hspace{-0.6mm}\hat{\boldsymbol{\lambda}}_{u|v}^{(l+1)m}\|_{\boldsymbol{P}_{d,uv}}^2\nonumber\\
&-\|\boldsymbol{c}_{uv}-\boldsymbol{A}_{u v}\hat{\boldsymbol{x}}_u^{(l+1)m}-\boldsymbol{A}_{v u}
\hat{\boldsymbol{x}}_v^{(l+1)m}\|_{\boldsymbol{P}_{p,uv}}^2\nonumber\\
&\hspace{-1mm}=g(k_1,u(k_1),v)\hspace{-0.4mm}-\|\hat{\boldsymbol{\lambda}}_{v|u}^{(l+1)m}\hspace{-0.4mm}-\hspace{-0.4mm}\hat{\boldsymbol{\lambda}}_{u|v}^{(l+1)m}\|_{\boldsymbol{P}_{d,uv}}^2\nonumber\\
&\hspace{3mm}-\|\boldsymbol{c}_{uv}-\boldsymbol{A}_{u v}\hat{\boldsymbol{x}}_u^{(l+1)m}-\boldsymbol{A}_{v u}
\hat{\boldsymbol{x}}_v^{(l+1)m}\|_{\boldsymbol{P}_{p,uv}}^2\nonumber\\
&\hspace{-1mm}= \Big[\boldsymbol{P}_{d,uv}\left(\hat{\boldsymbol{\lambda}}_{v|u}^{lm}-\hat{\boldsymbol{\lambda}}_{v|u}^{(l+1)m}\right)\hspace{-0.5mm}+\hspace{-0.5mm}\boldsymbol{A}_{v u}
\left(\hat{\boldsymbol{x}}_v^{(l+1)m}\hspace{-0.4mm}-\hspace{-0.4mm}\hat{\boldsymbol{x}}_v^{lm}\right)\nonumber\Big]^T\hspace{-0.8mm}\nonumber\\
&\cdot\hspace{-1mm}\left(\hat{\boldsymbol{\lambda}}_{u|v}^{(l+1)m}\hspace{-0.4mm}-\hspace{-0.4mm}\boldsymbol{\lambda}_{u|v}^{\star}\right)\hspace{-0.5mm}+\hspace{-0.5mm}\Big[\boldsymbol{P}_{p,uv}\boldsymbol{A}_{v u}
\Big(\hat{\boldsymbol{x}}_v^{(l+1)m}\hspace{-0.4mm}-\hspace{-0.4mm}\hat{\boldsymbol{x}}_v^{lm}\Big)\hspace{-0.5mm}+\hspace{-0.5mm}\hat{\boldsymbol{\lambda}}_{v|u}^{lm}\hspace{-0.4 mm} \nonumber\\
&\hspace{1mm}-\hspace{-0.6mm}\hat{\boldsymbol{\lambda}}_{v|u}^{(l\hspace{-0.3mm}+\hspace{-0.3mm}1)m}\Big]^T\hspace{-0.8mm} \boldsymbol{A}_{u v}\hspace{-0.8mm}\left(\hat{\boldsymbol{x}}_{u}^{(l\hspace{-0.3mm}+\hspace{-0.3mm}1)m}\hspace{-0.6mm}-\hspace{-0.6mm}\boldsymbol{x}_{u}^{\star} \right)\hspace{-0.6mm}-\hspace{-0.6mm}\|\hat{\boldsymbol{\lambda}}_{v|u}^{(l+1)m}\hspace{-0.6mm}-\hspace{-0.6mm}\hat{\boldsymbol{\lambda}}_{u|v}^{(l+1)m}\|_{\boldsymbol{P}_{d,uv}}^2\nonumber\\
&\hspace{4mm}-\|\boldsymbol{c}_{uv}-\boldsymbol{A}_{u v}\hat{\boldsymbol{x}}_u^{(l+1)m}-\boldsymbol{A}_{v u}\hat{\boldsymbol{x}}_v^{(l+1)m}\|_{\boldsymbol{P}_{p,uv}}^2\label{equ:optCon_gen_asyn2} \\
&\hspace{-1mm}= d_{uv}^{l+1}\qquad u<v, \label{equ:optCon_gen_asyn3}
\end{align}
where $d_{uv}^{l+1}$ is given by (\ref{equ:d_asyn}), of which the derivation is similar to that for $d_{i|j}^{k+1}$ in (\ref{equ:d_syn}). The term $u(k_1)$ in (\ref{equ:optCon_gen_asyn2}) is simplified as $u$ since we already assume that at iteration $k_1$, node $u$ is activated. The quantity $d_{uv}^{l+1}$ is a function of $m$ and $l$ instead of $k_1$. Finally, combining (\ref{equ:optCon_gen_asyn1}) and (\ref{equ:optCon_gen_asyn3}) produces (\ref{equ:optCon_gen_asyn}). 

\ifCLASSOPTIONcaptionsoff
  \newpage
\fi

\bibliographystyle{IEEEtran}
\bibliography{sigProcessing}

\end{document}